\RequirePackage{etex}

\documentclass[preprint,times]{elsarticle}
\usepackage{amsmath, amsthm, amssymb, mathrsfs}
\usepackage{dcpic, pictex, pinlabel, setspace, rotating}
\usepackage{array, verbatim, bbm, faktor}
\usepackage{graphicx}
\newtheorem{theorem}{Theorem}[section]

\newtheorem*{claim}{Claim}

\theoremstyle{definition}

\newtheorem*{protocol}{Protocol}

\theoremstyle{remark}

\newtheorem{example}[theorem]{Example}
\newtheorem{note}[theorem]{Note}

\numberwithin{equation}{section}
\numberwithin{figure}{section}

\newcommand{\innerprod}[2]{\langle {#1} , {#2} \rangle}

\newcommand{\R}{{\mathbb R}}

\newcommand{\mcH}{{\mathcal H}}

\newcommand{\C}{{\mathbb C}}
\newcommand{\bbar}{\overline}

\newcommand{\til}{\widetilde}

\def\hline{\bigskip\hrule\bigskip}  % temporary def
\newcommand{\ket}[1]{\left|{#1}\right\rangle}
\newcommand{\bra}[1]{\left\langle{#1}\right|}

\journal{Annals of Physics}

\begin{document}

\begin{frontmatter}

\title{Topological Phases: An Expedition off Lattice}

\author[sba]{Michael H. Freedman}
\author[zrh]{Lukas Gamper}
\author[zrh,tor]{Charlotte Gils}
\author[zrh]{Sergei V. Isakov}
\author[sba]{Simon~Trebst}
\author[zrh]{Matthias Troyer}

\address[sba]{Microsoft Station Q, University of California, Santa Barbara, CA 93106}
\address[zrh]{Theoretische Physik, ETH Zurich, 8093 Zurich, Switzerland}
\address[tor]{Samuel Lunenfeld Research Institute, Mount Sinai Hospital, 600 University Ave, Toronto, ON M5G 1X5, Canada}

\date{\today}

\begin{abstract}
Motivated by the goal to give the simplest possible microscopic foundation for a broad class of topological phases,
we study quantum mechanical lattice models where the topology of the lattice is one of the dynamical variables. 
However, a fluctuating geometry can remove the separation between the system size 
and the range of local interactions, which is important for topological protection and ultimately the stability of a
topological phase. In particular, it can open the door to a pathology, which has been studied in the context of 
quantum gravity and goes by the name of `baby universe', Here we discuss three distinct approaches to suppressing
these pathological fluctuations. 
We complement this discussion by applying Cheeger's theory relating the geometry 
of manifolds to their vibrational modes to study the spectra of Hamiltonians.
In particular, we present a detailed study of the statistical properties of loop gas and string net models 
on fluctuating lattices, both analytically and numerically. 
\end{abstract}

\begin{keyword}
%% keywords here, in the form: keyword \sep keyword
statistical mechanics \sep topological phases
%% PACS codes here, in the form: \PACS code \sep code

%% MSC codes here, in the form: \MSC code \sep code
%% or \MSC[2008] code \sep code (2000 is the default)

\end{keyword}

\end{frontmatter}

%------------------------------------------------------------------------------

\section{Introduction}

Topological states of matter have become a major topic in condensed matter theory and experiment. 
Quite simple microscopic Hamiltonians, e.g. those of  fractional quantum Hall (FQH) systems, are known to harbor emergent topological phases 
but often determining the correct effective low-energy theory is difficult: it becomes a delicate matter of energetics once one moves beyond idealized short ranged interactions. Another way of idealizing interactions comes in the form of quantum lattice models. In recent years lattice models have become another, complementary source of topological phases. Examples include the well-studied toric code model \cite{ToricCode} and the Kitaev honeycomb model \cite{Kitaev},
which realizes the Ising topological quantum field theory (TQFT) in a controlled perturbative regime. A more general class of lattice models was introduced by Levin and Wen \cite{LW} -- ``Turaev-Viro" \cite{TV} in the math literature -- which are rich sources of exactly solved achiral (``doubled") theories.

There is, however, a trade-off in using these models. In the exactly solved LW lattice models we know everything, but the Hamiltonians look artificial: finely tuned, 12-body interactions without the comfort of familiar ``kinetic" and ``potential" terms. In FQH the Hamiltonian is simple and natural, but the connection to the (or several competing) topological phases is not obvious. The motivation of this paper is the search for a middle ground where a simple local Hamiltonian can be clearly identified with a unique topological phase. Our chief innovation is to treat the lattice itself as a {\em dynamic} variable so that there is no fixed lattice underlying the Hilbert space, hence the words ``off lattice" in the title. Our search for this middle ground is presented as a travel log, with some surprises, disappointments, and discoveries. Along the way we came to better appreciate what exactly a fixed lattice is good for and what adaptations its absence requires. Briefly, a lattice model supplies two length scales, the lattice scale $a$ and the length of period $L$, where topological protection comes from an error scaling: $\epsilon\approx e^{-{\rm const} \cdot L/a}$. Protecting quantum information without this ratio of scales, e.g. when going off lattice, is a key challenge.

Levin-Wen models place degrees of freedom (``labels") on the edges of a trivalent graph dual to a fixed triangulation $\Delta$ of a surface $\Sigma$, where the labels come from a fusion category. 
The labeled graph is called a ``string net". There is a mathematical fact encouraging us to leave the concept of a {\em fixed} lattice behind. Starting with a consistent $F$ symbol (i.e. one obeying the pentagon relations) for a fusion category 
\begin{equation}
  \parbox{2cm}{
    \scalebox{0.25}{
      \includegraphics{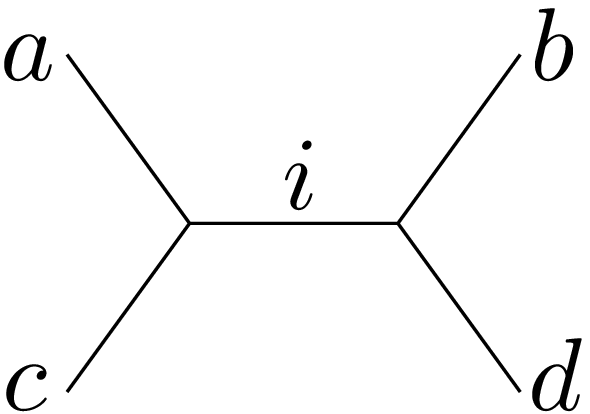}
    }
  }
  = \sum_j F^{abi}_{cdj}
  \parbox{2cm}{
    \scalebox{0.25}{
      \includegraphics{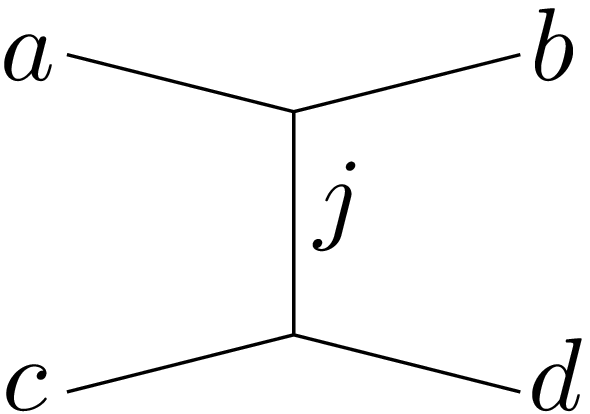}
    }
  }
  \label{fcat}
\end{equation}
and a closed surface $\Sigma$, the vector space of admissibly labeled string nets on $\Sigma$ modulo isotopy and $F$ is canonically isomorphic to the LW ground state Hilbert space $V(\Sigma)$ built from a fixed dual triangulation. Clearly (\ref{fcat}) is a simpler starting point than the LW Hamiltonian, but there is a serious problem: no single lattice (we use the word ``lattice" here to mean the dual structure to a triangulation of $\Sigma$) will contain simultaneously bonds forming the ``stick figures" on both the left and right hand sides of (\ref{fcat}); the $F$-move necessarily modifies the lattice.

In fact the LW Hamiltonian arises from a composition of such $F$-moves which --- like a
perturbation through excited states --- begins and finally ends on the same
lattice. For the honeycomb lattice, the LW Hamiltonian is a 12-body term, arising from 6 $F$-moves. Our idea here is to construct an environment where
the $F$-move itself may be written directly into the Hamiltonian forcing the underlying lattice to fluctuate. 

This paper starts by reviewing the problems we encounter when  going off-lattice in Sec. \ref{sec:hazards} and then present various ways to achieve off-lattice topological models in Sec. \ref{sec:models} before concluding with a summary of the lessons learned from this off-lattice expedition in Sec. \ref{sec:conclusions}. The appendices contain a number of results that are interesting beyond the main subject of this paper. \ref{Appendix:Spectrum} 
 proves a discrete version of Cheeger's theorem appropriate graphs with weights on edges and vertices extending work of Chung \cite{Chung}. We use it to derive bounds on the gaps of graph Laplacians and other local Hamiltonians. \ref{Appendix:OuterPlanar} derives an upper bound for the scaling of the gap of the graph Laplacian on outer planar triangulations, following an idea developed with Oded Schramm.  \ref{Appendix:C} presents exact results for an off-lattice version of the toric code model, and finally \ref{Appendix:Numerics} presents numerical results for the gap of the Graph Laplacian and Cheeger's constant for trees, triangulations of the sphere and outer planar triangulations.

\section{Off Lattice Hazards}
\label{sec:hazards}

Let us begin by cataloging three hazards which await us off-lattice. 

\subsection{Baby Universe}\label{subsection_baby_universe}

The first hazard is well known in Euclidean quantum gravity. It is called ``minbus''
or ``baby universe'' \cite{AR}. It refers to the fact that if a triangulation (or any
other polyhedral decomposition) is chosen uniformly at random\footnote{or by any other local formula.} for a sphere
(or other closed surface of fixed genus $g$) among all triangulations of a
fixed number $n$ of triangles, it is likely that there will be a short dual loop
containing numerous vertices on both sides (or in the case of surfaces with
$g>0$ there will also be short nonseparating loops). We measure the length of
a (dual) loop simply as the number of edges it crosses. This can be
formalized by saying that Cheeger's \cite{Cheeger} isoperimetric constant $\longrightarrow 0$
when $n\longrightarrow \infty$, almost always
\begin{equation}
k=\min_{\text{separating dual loops } \gamma} \frac{length(\gamma)}
  {\min[area(S), area(\bar{S})]},
\end{equation}
where $S$ and $\bar{S}$ are the components of the surface $\Sigma$ minus
$\gamma$, and $area(S)$ ($area(\bar{S})$) is the number of vertices (or sites)
in $S$ ($\bar{S}$).

It is known that typically a triangulated sphere has
$k\approx O(\frac{\log n}{n})$.
There is in fact an asymptotic formula \cite{Schaeffer} for the number
$\#_g(n)$ of isomorphically distinct divisions of a closed genus $g$ surface
into squares with the condition that the dual graph has no odd cycles:
\begin{equation}
\#_g(n)=12^n n^{\frac{5}{2} (g-1)}=12^n n^{-\frac{5}{4}\chi},
\label{eq:nd}
\end{equation}
where $\chi$ is the Euler number.
It is only a simplifying technicality to treat these quadrangulations rather
than triangulations --- similar asymptotics should apply but with 12 replaced
by some other less convenient base; the $n^{-\frac{5}{4}\chi}$ is universal.

An easy application of Eq.~(\ref{eq:nd}) (we thank Gilles Schaeffer for bringing
this to our attention) is that it is possible to estimate the fraction $f_4(n)$
of genus $g$ surfaces that are divided by a separating curve of length four
into surfaces of genus $g_1$ and $g_2$, $g_1,g_2>0$. The method is to remove
a random square from random surfaces of  $g_1$ and $g_2$, glue the results
together and count how many ways this is possible as a fraction of surfaces of
genus $g=g_1+g_2$. The result is:
\begin{equation}
f_4(n)\approx n^{-1/2}.
\end{equation}

To summarize: constant size bottle necks are algebraically likely and
logarithmic bottlenecks are virtually assured. This appears to be very
unfavorable for the protection of topological information. When working with
lattice models, we are used to error rates appropriate to tunneling problems
like $e^{-\text{const}\cdot L}$, where $L$ is the linear dimension of the lattice as a multiple of the lattice constant.
If bottlenecks reduce $L$ to constant or even logarithmic size, the protection
disappears. This is the first issue to come to terms with when considering uniformly random triangluations (URT).

For quite a different reason, these bottlenecks and baby universes have been
unwelcome also in quantum gravity. There one seeks a Hamiltonian function on
triangulation which concentrates near Euclidean flat space but still allows
a liquid of possible universes. Forty years of effort have failed to find such
a phase even in dimension 2 (at least in the homogeneous setting), for a review see \cite{AR}. It is possible
to form a ``branching polymer universe'' of Hausdorff dimension $2$, 
to perturb about a single rigid Euclidean crystal, 
or a ``collapsed phase'' of infinite Hausdorff dimension, 
but a nearly flat, yet liquid, phase, or even a critical point, has been elusive. However, in the last dozen years progress has been
made by breaking the symmetry between space and time and allowing only 
triangulations appropriately
foliated by space-like leaves \cite{AR}. The approach is called ``causal
dynamical triangulation'' (CDT) and has been shown numerically to
provide ``birth control'' \cite{AJL} --- there are parameter regimes, called the ``C phase", with no baby
universes in which the space-like leaves are on average nearly Euclidean of the desired dimension. In Section 2.3 we describe an approach to building a $(1+1)+1$ dimensional model for a 2+1-dimensional anyonic system in which the 2 spatial dimensions are broken into a 1+1 pair to exploit the favorable statistical geometry of 1+1 CDTs.

% insert 1

\subsection{Gapless Modes}\label{subsection_gapless_modes}

The second issue with URT is the mixing time. We
analytically estimated the Cheeger isoperimetric constant $h$ in a toy model
of surface triangulations called ``outer planar'' triangulations. In this
context we show analytically $h \preccurlyeq n^{-\frac{1}{2}}$, whereas in our numerical
study (presented in \ref{Appendix:N-gons})
the first eigenvalue $\lambda$ of the graph Laplacian (the ``graph''
has vertices outer plane triangulations and edges plaquette flips between these)
goes like $\lambda\approx n^{-2}$ (this translates to the first eigenvalue
$\tilde{\lambda}$ of the graph incidence matrix scaling like
$\tilde{\lambda} \approx n^{-1}$ since $\lambda\approx\tilde{\lambda}{n^{-1}}$ in our models\footnote{
   We need to remove an exponentially small number of states from the Hilbert space to obtain this scaling.
   For a detailed discussion see \ref{Appendix:C} and \ref{Appendix:Numerics}.}). Cheeger like inequalities show (see
\ref{Appendix:Spectrum}): 
\begin{equation}
2h \succcurlyeq \lambda \succcurlyeq \frac{h^2}{2}\,, 
\end{equation}
or $$n^{-1}\succcurlyeq h \succcurlyeq n^{-2} \,.$$
We believe the truth is near the high end $h\approx n^{-1}$, and that on the sphere $\lambda\approx{n^{-1.75}}$.

As a further probe of the spectrum, we studied  the dynamics of string nets on the 2-sphere 
(see \ref{Appendix:NumericsSphere})
and also observed
a mixing time $\approx n^2$, i.e. $\lambda \approx n^{-2}$, and
$\tilde{\lambda}\approx n^{-1}$. String nets are dual to triangulations but slightly more flexible, e.~g. a closed loop in a string net is permitted whereas the usual definition of triangulation does not permit a triangle to be glued to itself.

All this confirms the findings of the quantum gravity community: the space of
random triangulations, quadrangulations, string nets, etc. on a surface will
mix algebraically fast but not so fast (which would need to be $O(n)$) so
that the $\tilde{\lambda}$ first eigenvalue of the incidence matrix is gapped
or, equivalently, that the Cheeger constant $k$ (of the graph-of-triangulations) is bounded away from zero.

% end of insert 1

\subsection{Local Distinguishability}\label{subsection_local_distinguishability}

A third problem was noticed when we studied multi loops on a periodic
honeycomb lattice. The chosen dynamics is that of the toric code \cite{ToricCode}, however, we
made the convention that if two multi loops were isotopic, deformable one to
the other, then they would be identified and represented by a single ket.
In \ref{Appendix:C}, we present data which shows that the trivial winding
sector can easily  be picked out from the other three by a ``local''
observation --- we count the number of ``leaves''  --- that is loops with no
smaller loops within. In retrospect, this is no surprise. Being in the trivial
sector allows the possibility of no essential loops whatever --- this possibility
permits more space on the lattice for leaves.

Local distinguishability is, of course, the death knell of topological
protection. A state which can be observed locally can be acted on by a local
operator. This is a third disturbing finding if we grant that leaves are to
be considered local structures. Since metrical notions have been temporarily banished, it is up to our intuition to reformulate the appropriate meaning of ``local opeator." Leaf detection has, in this context, as good of a claim to being local as does any operator.

%------------------------------------------------------------------------------

\subsection{Work-arounds for off Lattice Troubles}

Of the three problems, (2) is the least concerning. Even if there are low energy metrical fluctuations
(one may dub them ``gravity waves"), they appear decoupled from topological degrees of freedom which can be encoded on each lattice. With respect to problems (1) and (3), there is a somewhat solipsistic solution 
to the apparent loss of topological
protection from (1) bottlenecks and (3) variations in leaf count. It is
simply to deny that this is a problem. Once kets of the Hilbert space are
isotopy classes (of triangulations or nets --- perhaps together with a particle
type labelling of the bonds) we have lost direct contact with any notion of
a position coordinate $\vec{x}$. Isotopy slides and stretches, so we no
longer know what is long or short or even where we are. This viewpoint leads
one to say there are no local operators at all and therefore topological
protection --- protection against local operations --- is tautological. But such
a view comes with a heavy price --- without  a position coordinate $\vec{x}$
correlation functions loose meaning and contact with the condensed matter
notions disappears. Consequently we will not take this path but rather
consider three distinct approaches all of which enforce flat Euclidean space as
the background, but in quite different ways and with quite different
results.

In summary, we find that there are plausible and even intriguing ways to model topological phases off lattice. The next step should be to identify a case where the model variables can be mapped to electron degrees of freedom.

%------------------------------------------------------------------------------

\section{Enforcing Flat Space: Crystalline, Liquid off Lattice Models, and CDT}
\label{sec:models}
%\section{Y: Liquid and Crystalline Off Lattice Models}\label{section_Y}

Topological phases of lattice models are known in the physics
literature from Levin and Wen \cite{LW} and in the math literature
from Turaev and Viro \cite{TV}.  We explore what happens in taking
such a model ``off lattice'' by including the underlying cell
structure or ``lattice'' among the dynamic variables.  Thus, our
Hilbert space $\mathcal{H}$ will be spanned by kets which are pairs
$|(\Delta, S) \rangle$ where $\Delta$ is a triangulation of the
surface $\Sigma$ with $n$ (fixed) triangles and $S$ is a ``labeling''
of the dual edges of $\Delta$.  We can equally well focus on the dual
string net and its variations, some of which may not be dual to actual
triangulations.  It is the dual edges which form our so-called
lattice, e.g.\ the dual edges form a honeycomb if $\Delta$ is the
standard triangulation of the plane by equilateral triangles.  The
labels are from a pivotal fusion category.  Two interesting examples
are with label set \{1, $\tau$\} of the Fibonacci theory Fib and fusion rule $\tau \otimes \tau =
1 \oplus \tau$, which when applied as labels on string nets yields the theory Dfib \cite{FFNWW}, or with
label set $\{1,x\}$ and fusion rule $x \otimes x = 1$, which gives the
toric code \cite{Kitaev}.  Rather than speaking in generalities, we give
our constructions in the former case.  They are easily extended to the
broader class.

In moving off lattice we can be timid or bold but as we have argued we
must find some way to tie our lattices to the Euclidian plane.  Our
Hamiltonian can charge energy for defects in a base lattice, say the
honeycomb, or it can treat all triangulations equally.  There is, of
course, an adjustable parameter connecting large to zero energy
penalty for defects.  However, decades of experience with 2D quantum
gravity (qg) models suggest that there is a single phase
transition from a phase with amplitudes clustered near the original
honeycomb (we call this crystalline) and a liquid of lattices whose
geometry is almost surely ``cactus shaped'' -- Cheeger constant $\sim
\frac{\log n}{n}$.  In the quantum gravity community the lack of an intermediated
phase was the cause of some despair, the cactus buds being called
``baby universes.''  As we remarked in the introduction, this problem may have been solved
in the quantum gravity context \cite{AJL,AR} by introducing an appropriate causal
structure. This is explored in Section 2.3, where 3D space is split into a radial ``space" coordinate $\rho$ and a periodic ``pseudo-time" coordinate $\theta$. But if one insists that 2D space be treated homogenously, then the baby universes must be faced.  We will do this but first let us explore moving only timidly
off the honeycomb.  There we find a gapped model which is conceptually
very simple (we think more simple than the LW model) but the price of
our timidness is that the gap is absurdly small, perhaps $\approx
\epsilon \left( \frac{1}{4} \right)^{54}$, where $\epsilon$ is the
energy scale of the individual terms.  In spite of the disappointingly
small gap, we next explain this model as it is a nice, controlled
context for stepping -- ever so slightly -- off lattice. A bolder step
will be taken later.

\subsection{Crystalline Case}\label{subsection_crystalline_case}

Our first Hamiltonian has the form $H_{qg} = H_{qg}^0 + \delta D$. $H_{qg}$ acts on the direct sum of fibers of a bundle of states over
the moduli space of string nets (thought of as metrics) on $\Sigma$, say a torus.  The terms of
$H_{qg}^0$ are of two types:
\begin{enumerate}
\item Fusion constraints; these are projectors acting within fibers
\item $F$-moves; these act between adjacent fibers
\end{enumerate}
The fibers are degrees of freedom on the edge set of any given string net.
$F$-moves define a connection linking these fibers together which, because
of the pentagon relation, trivialize the subbundle satisfying fusion constraints.

The term $\delta D$ is an energy penalty which charges energy
$\delta$ for each pair of (5-gon, 7-gon) pairs created by an $F$-move (see Figure 2.4a). The number of such defects is counted by the operator $D$.  The expected
Levin-Wen 12-body plaquette term is {\em not} directly included but we
will show that it arises at high order by considering the process which virtually breaks a $\tau$-labeled string, resulting in a pair of ``electric" excitations\footnote{In the notation of \cite{FFNWW} the electric pair may be either $(\tau \otimes 1, \tau \otimes 1)$ or $(1 \otimes \tau, 1 \otimes \tau)$.}
costing energy $2\epsilon$, see \cite{Freedman}.
We consider two
triangulations $\Delta$ and $\Delta'$ (and their dual nets (or
``lattices'') $N$ and $N'$) to be equivalent if they are isotopic on
$\Sigma$, i.e. if we can slide one to the other.  Later, we set
$\delta$ to zero to obtain a lattice liquid, then add a string tension
term and also tie each net to $\Sigma$ in a fixed way (unrestricted isotopy will no longer be permitted);
but for now we study the crystalline case and need not distinguish isotopic nets.  The dynamics on the set of
nets $\mathcal{N}_n$ dual to $n$-vertex triangulations is shown in
Fig.~\ref{fig_move}.

\begin{figure}[htpb]
\begin{center}
\includegraphics[scale=0.5]{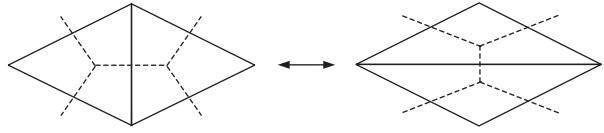}
\end{center}
\caption{H-I move, called $F$-move when coefficients are added as in (1.1).}
\label{fig_move}
\end{figure}

Now let us define $H_{qg}^0$ more precisely.  First, it enforces
fusion rule terms at each vertex of each net $N^i$ by penalizing the
illegal Fibonacci fusion (see Fig.~\ref{fork}) and its symmetries. 
\begin{figure}[ht]
\labellist \normalsize\hair 2pt

  \pinlabel $\tau$ at -10 240
  \pinlabel $1$ at -10 0
  \pinlabel $1$ at 320 120

\endlabellist
\centering
\includegraphics[scale=0.2]{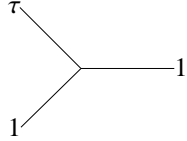}
\caption{Illegal Fibonacci fusion.} \label{fork}
\end{figure}

Second, it contains terms between states of
adjacent nets $N$ and $N'$ which enforce the unitary $F$-symbol
$\displaystyle \left| \begin{array}{lr} \tau^{-1} & \tau^{1/2} \\
    \tau^{1/2} & -\tau^{-1} \end{array} \right|$, $\tau = \frac{1 +
  \sqrt{5}}{2}$. Let $v$, $w$ be the normalized states of $H$, shown in
Fig.~\ref{fig_states}.  The second terms of $H_{qg}^0$ are of the form
$(id - |v\rangle\langle v|)$ and $(id - |w\rangle\langle w|)$.  
\begin{figure}[htpb]
\labellist \normalsize\hair 2pt

  \pinlabel $v=$ at 15 220
  \pinlabel $w=$ at 15 75
  \pinlabel $-\tau^{-1}$ at 215 220
  \pinlabel $-\tau^{1/2}$ at 370 220
  \pinlabel $-\tau^{1/2}$ at 215 75
  \pinlabel $+\tau^{-1}$ at 370 75

\endlabellist
\centering
\includegraphics[scale=0.325]{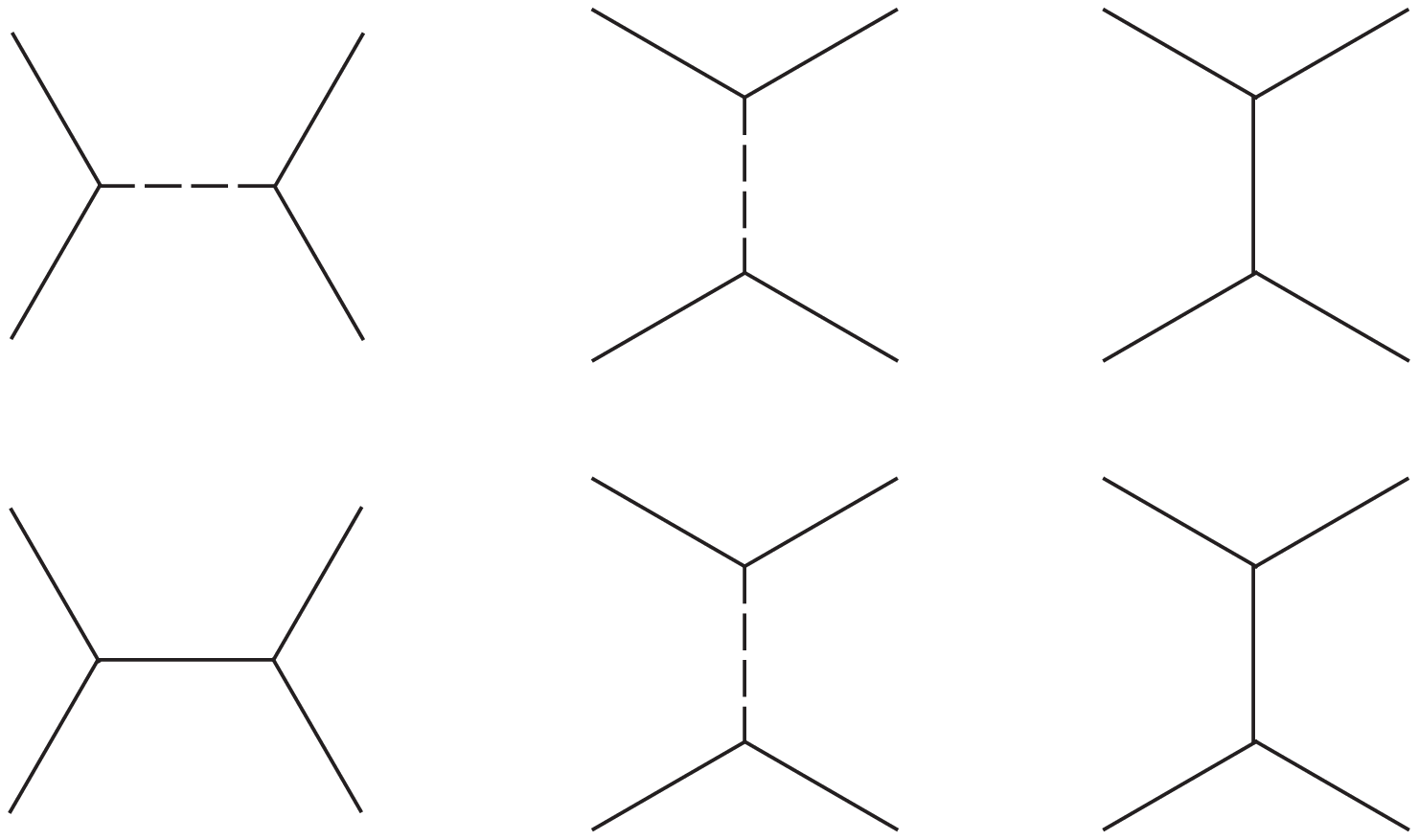}
\caption{The states $v$ and $w$. Solid lines carry the $\tau$ particle label and
dotted lines the trivial label.} \label{fig_states}
\end{figure}

We discuss the spectrum of $H_{qg}^0$ first.  $H_{qg}^0$ is positive
semi-definite and its ground state manifold consists of the states
$\psi$ with $\langle \psi | H_{qg}^0 | \psi \rangle = 0$.  Such a wave
function $\psi$ is completely determined via the $F$-symbols by its
restriction to a sample net $N^0$, e.g. a honeycomb.  (Importantly, $\psi$
is not {\em over} determined (frustrated) since the $F$-symbol
satisfies the famous  pentagon equations.)  The ground state manifold
may be classified according to the number of magnetic particles\footnote{In notation of \cite{FFNWW} $m=\tau\otimes\bar\tau$.} $m$ (of which, in our example system, there is only one
type).  Since we have only imposed fusion and $F$-moves there is no
energy penalty for $m$ charges on plaquets, provided that, unlike in the Levin-Wen model, they are allowed to roam ergodically according to
the moves ($F$) which link adjacent nets.  The magnetic charges on
$N^0$ can return arbitrarily permuted, so the only zero energy
(unfrustrated) states with $j$-magnetic charges, $j \geq 2$, are the
ones that have equal amplitude for all positions of the $j$ charges
(on all $n$-vertex nets). In contrast to the LW model magnetic charges are localized gapped excitations, the $F$-moves here delocalize them and lead to a gapless continuum of magnetic charges above the ground state. In addition, there is also a continuum of gapless ``gravity waves,'' or
phase oscillation across the (not very tightly bound) graph
$\mathcal{N}_n$ (see section \ref{Appendix:OuterPlanar}). This is analogous to coexisting gapless magnon and phonon excitations in a quantum magnet.
$\mathcal{N}_n$ is regarded as an abstract graph with the
triangulations $\Delta_n$ (or nets $N$) as vertices and edges given by the move
shown in Fig.~\ref{fig_move}. The appropriate vertex weighting of
$\mathcal{N}_n$ (see \ref{Appendix:Spectrum}) is uniform until the
$\delta D$ term is added.  As explained in 1 (``Introduction") and
\ref{subsection_liquid_case} (``Liquid Case'') below it is believed
that $\lambda_1(\mathcal{N}_n) \approx \frac{1}{n}$.

Adding the term $\delta D$ induces frustration which causes the ground
state wave function $\psi_0$ to concentrate near the original
honeycomb states.  When $\delta > \text{const}\cdot \gamma$, some
$\text{const} \approx 4$ and $\gamma$ the energy scale of the $F$-symbol,
kinetic considerations are overwhelmed and there will be a phase
transition to exponentially small fluctuations around the honeycomb
configurations. This concentration alters the weights 
(see \ref{Appendix:Spectrum}) on the vertex set $V(\mathcal{N}_n)$ 
and is expected to gap out the gravity waves.

\begin{figure}[htpb]
\labellist \small\hair 2pt
  \pinlabel $5$ at 390 292
  \pinlabel $5$ at 200 292
  \pinlabel $7$ at 290 210
  \pinlabel $7$ at 290 388
  \pinlabel $\text{Second $I-H$ move here.}$ at 800 210
  \pinlabel $\text{$\xleftarrow{\hspace*{1.1cm}}$}$ at 465 210
  
\endlabellist
\centering
\includegraphics[scale=.2]{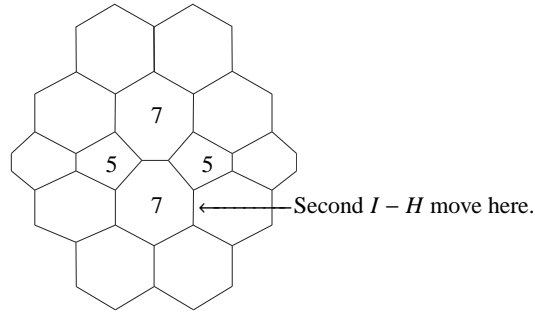}
\caption{One $I-H$ move creates a pair of $(5-$gon, $7-$gon$)$ pairs costing energy $\delta$.} \label{fig_HD}
\end{figure}

\begin{figure}[htpb]
\labellist \small\hair 2pt
  \pinlabel $5$ at 460 175
  \pinlabel $5$ at 200 280
  \pinlabel $7$ at 360 105
  \pinlabel $7$ at 290 375
  
\endlabellist
\centering
\includegraphics[scale=.2]{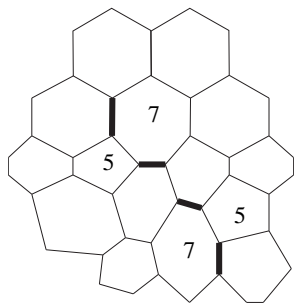}
\caption{A second $I-H$ move separates the two pairs. A third $I-H$ move along any of the four bold bonds (Figure 3.4) causes a $(5,7)$ pair to further propagate. Delocalizing the defect at $\delta=4\gamma$ in lead order in perturbation theory in $\gamma/\delta$.} \label{fig_HD2}
\end{figure}

Treating a pair of $(5,7)$-gons as the fundamental excitation with cost $\delta$, we see that the hexagonal crystal melts (at first order in perturbation theory) for ``kinetic energy" $\gamma$ associated with the $F$-move satisfying $\gamma>\frac{\delta}{4}$.

Now consider a virtual excitation (of energy cost $=\lambda$) which pulls an
electric pair (say $(\tau \otimes 1, \tau \otimes 1)$) out of the
vacuum.  Because of the {\em nontrivial} mutual statistics between the
magnetic ($\tau \otimes \tau$) and electric ($\tau \otimes 1$)
excitations, a frustration arises which increases the cost of the
electric pair $\psi_j^{e,e^*}$ in the presence of $j$ magnetic
particles.  For small $j$ the effect is roughly linear: 
\begin{equation}
\langle
\psi_j^{e,e^*} | H_{qg}^0 | \psi_j^{e,e^*} \rangle - \langle
\psi_0^{e,e^*} | H_{qg}^0 | \psi_0^{e,e^*} \rangle \approx j \alpha
\end{equation}
for some $\alpha > 0$ and where we have set $\langle \psi_0^{e,e^*} |
H_{qg}^0 | \psi_0^{e,e^*} \rangle =2\epsilon$.

Here $\alpha = \gamma/54$ is the energy scale $\gamma$ of the
$F$-symbol constraint divided by the number of $F$-moves required to
take one plaquette $B$ of $N_0$ around a neighbor $A$ and across an
``electric string'' (see Fig.~\ref{fig_ih_moves}). This $1/n$ scaling
of $\alpha$ mirrors that of the ground state energy of a one dimensional ferromagnet on a system
of length $n$ with twisted boundary conditions.

\begin{figure}[htpb]
\labellist \small\hair 2pt
  \pinlabel $A$ at 75 142
  \pinlabel $B$ at 75 85
  \pinlabel $1$ at 261 170
  \pinlabel $2$ at 310 143
  \pinlabel $3$ at 310 87
  \pinlabel $4$ at 261 58
  \pinlabel $5$ at 213 84
  \pinlabel $6$ at 204.5 98
  \pinlabel $7$ at 201 114
  \pinlabel $8$ at 204.5 130
  \pinlabel $9$ at 213 144
  \pinlabel $B$ at 215 172
  \pinlabel $A$ at 215 190
  \pinlabel $B'$ at 233 181
  \pinlabel $A$ at 288 191
  \pinlabel $B$ at 288 172
  \pinlabel $B$ at 286 66
  \pinlabel $B$ at 216 66
  \pinlabel $B$ at 323 129
  
\endlabellist
\centering
\includegraphics[scale=1]{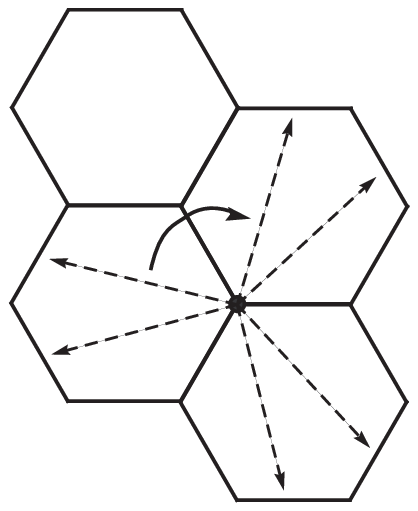}
\hspace{1cm}
\includegraphics[scale=0.3]{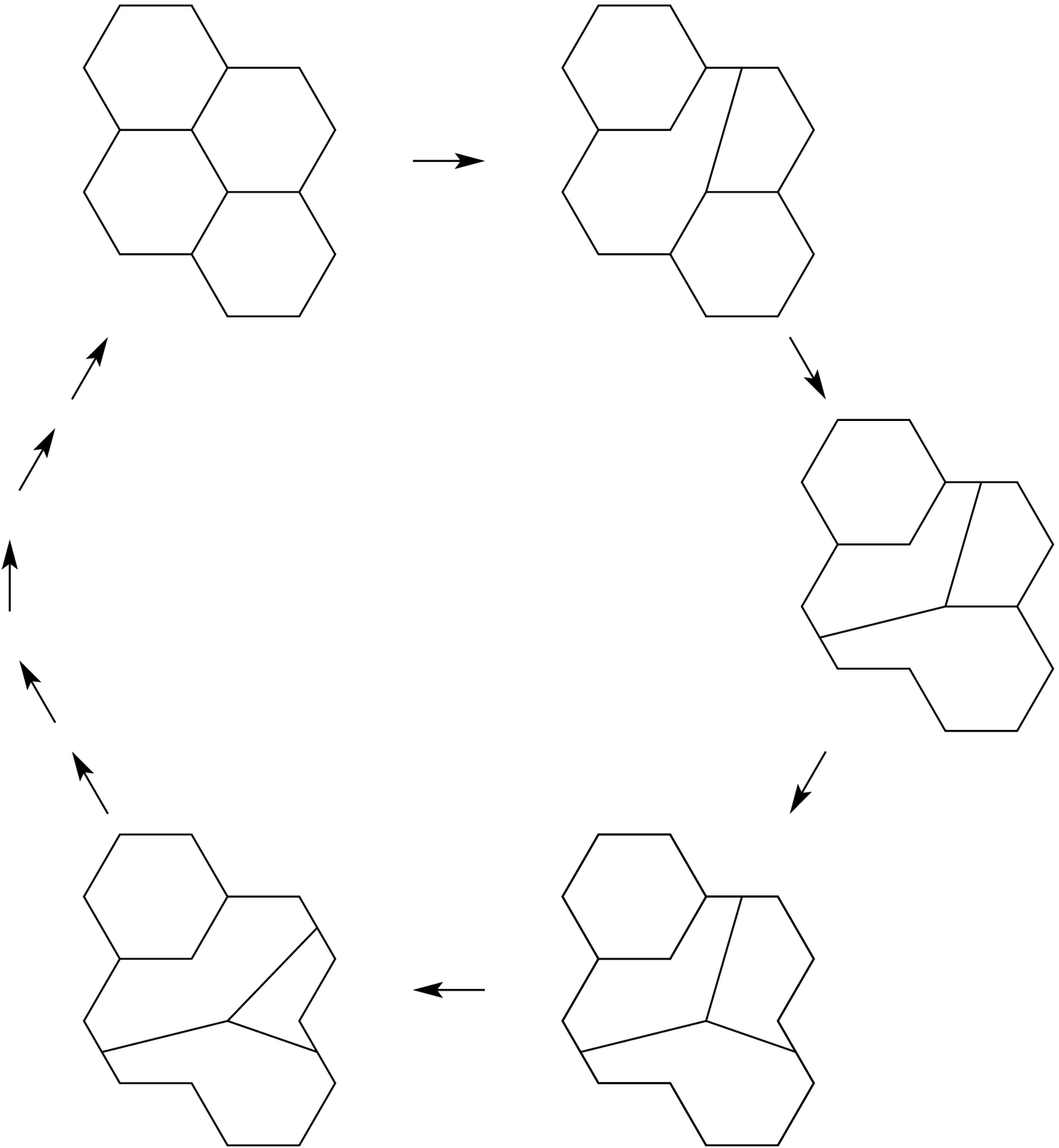}
\caption{9 $I$-$H$-moves (or $F$-moves at the Hamiltonian level)
  rotate the three bonds meeting the ``dot'' $+120^\circ$.  This
  rotates hexagon $B$ $-60^\circ$ around its neighbor $A$ from $B$ to $B'$.  54 moves
  complete the circuit.} \label{fig_ih_moves}
\end{figure}

The splitting which separates the ``true'' $j=0$ vacuums now show up at
56th order in perturbation theory.  ($56 = 54 + 2$, the 54 counts the
steps in Fig.~\ref{fig_ih_moves} to move $B$ around $A$ and the 2
comes from first creating then removing the ``electric'' pair.)  The
true vacuum has its energy lowered schematically\footnote{Slightly
  more accurately by $(\text{combinatorial factors)} \cdot
  \lambda^2\gamma^{54}(\delta+2\epsilon)^{-48}(2\epsilon)^{-7}$.} by
$\displaystyle \frac{\gamma^{56}}{(\delta + 2\epsilon)^{55}}$, whereas
the $j$ vacuums have an energy reduction of $\displaystyle
\frac{\gamma^{56}}{(\delta + 2\epsilon + j\alpha)^{55}}$.  These
numbers are each $\approx 55$th powers of a small and a somewhat
smaller number, respectively.  One may say that perturbation theory
predicts, in some regime, a definite splitting off of the true vacuum
which, although vanishingly slight, is constant in system size.  Thus
in summary, the timid ``crystalline'' off lattice approach succeeds in
principle but may be judged physically useless, because the gap will
be tiny.

\subsection{Liquid Case}\label{subsection_liquid_case}

Let us now move to the other extreme and drop the $\delta D$ term by
setting $\delta = 0$.  Now our kets are over a liquid of ``lattices''
or ``nets'' familiar in the quantum gravity literature.  Let us
summarize what is known about the statistics of these nets through
theoretical and numerical study (see \ref{Appendix:OuterPlanar}, \ref{Appendix:Numerics}, \cite{AJL}, and \cite{AR}).  Given uniform weight, a weight proportional
to total Gauss (= scalar) curvature, a topological quantity in
dimension 2, or any other known local weight which does not enforce a
``crystal,'' the geometry is cactus-like, with many budding or
``baby'' universes.  These correspond to Cheeger constant $\approx
\frac{\log n}{n}$ (using the combinatorial weights), $n =
\#$triangles.  That is, bottlenecks of size $\log n$ are common.  In
fact, the probability of a bottleneck of constant size is $\approx
n^{-\frac{1}{2}}$, i.e.\ only algebraically small.  In a related vein,
studying mixing times suggests (see \ref{Appendix:Numerics}) that
$\lambda$ ($= \lambda_1$ of $L$, see \ref{Appendix:Spectrum})
decays as approximately $n^{-1.75}$ when the nets are weighted
uniformly. We find a similar exponent  for the related case of
multi-loop rather than net dynamics.

This seems to present us with two problems:
\begin{enumerate}
\item Gapless gravity waves
\item Loss of a length scale
\end{enumerate}

The first turns out, by itself, not to be a serious problem. It is
actually quite interesting to have a simple mathematical model
which manifests gapless modes living side by side with protected
topological degrees of freedom.  In the context of FQH states, if the model is taken to
be sufficiently comprehensive to include lattice ions, then surely
their phonons are also an example of this phenomenon.  On the other
hand, the loss of length scale is inherent in declaring kets to be
{\em isotopy classes} of labeled nets is a serious problem.  We no
longer know if a bond is long or short, straight or wiggly. We view
with concern the loss of combinatorial protection conveyed by a large
regular lattice.  Recall that on an $L \times L$ torus mixing of
topological sectors occurs via tunneling along a Wilson loop of length
$L$ and will be suppressed by a factor of $e^{-\text{const}\cdot L}$.
As noted in the introduction, if bottlenecks cause $L$ to be replaced by $\log L$, or
even worse a constant, then the exponential protection disappears. 

To deal with this problem (2), we introduce a fixed fine-scale lattice
on the surface $\Sigma$ (perhaps writing $\Sigma$ as an $L \times L$
torus) and regulate our nets to lie within this fine lattice.  The
nets $\mathcal{N}_n$ still are restricted to $n$ vertices,
$n=\text{small constant}(L^2)$, but now their detailed position in
$\Sigma$ is {\em pinned} as part of the data of a ket $|(\Delta, S)
\rangle$; $\Delta$ is regarded now as a specifically located, or
pinned, $n$-vertex net in $\Sigma$. We will need to impose something that acts like ``string
tension'' that prevents the net bonds from becoming too long as
measured in the underlying fine $L\times L$ grid. This prevents short essential loops and so avoids baby universes. As explained below, the bonds become ``virtual", only their end points are precisely located. String tension can be simulated by establishing a hard energy penalty term $\omega B$ in $H$,
which charges energy $\omega$ for net bonds longer than $\ell_{\rm max}$ grid
bonds (counted by the operator $B$). Alternatively, a harmonic string tension can be imposed.

Technically the simplest way to incorporate our pinning and string
tension terms is to alter the basic Hilbert space on which the
Hamiltonian is defined.  Begin with a fine lattice of sites on the
surface (such as a torus) and as kets take all pairs: (a bond indexed
by two sites no more than $\ell_{\rm max}$ steps apart and thought of as
joining the sites, a label on the bond). Note that the precise physical
placement of the bond is \textit{not} chosen to be part of the data defining a ket.\footnote{This simplifies detailed balance for the $I$-$H$-moves.}
One may say that the string net is "imbeddable" --- according to certain rules --- 
but not "imbedded." The bonds at this level are "virtual." The labeling just
mentioned is from the appropriate set of quantum group
representations –-- as is usual --– \{1,$\tau$\}. The fusion
constraints now specify that the virtual bonds first form a
trivalent string net and  second that the three labels at any
juncture obey the algebraic fusion rules appropriate to the system of
quantum group representations being used.  An additional ``isotopy'' terms shifts the location of a vertex within the underlying lattice, provided all distance constraints are satisfied. The $F$ symbol
applies to recoupling virtual bonds.

Let us explain why equally weighted pinned nets are gapless
under these local moves. The situation is only a slightly more
global version of ``the space of all arcs transversing a rectangle''
$=: X$.  A typical arc will be nearly dense -- it will come within a
constant distance of a positive fraction of lattice points.
To define a bottleneck or ``Cheeger cut'' on this space of arcs,
consider the mid point $m$, in terms of arc length, of every arc.
Let $U(L)$ be the set of arcs for which $m$ lies in the upper(lower)
half of $X$. The ``cut'' is $U \bigcap L$.  Since the probability
density if nearly uniform for $m$ in $X$, the Cheeger constant
satisfies $k \preccurlyeq \frac{1}{L}$. Thus by
\ref{Appendix:Spectrum} $\frac{2}{L}\succcurlyeq\lambda$. The nets
will still be gapless after pinning.

Pinning the net restores the $e^{-\text{const}. L/a}$ scaling for
tunneling of quasi-particles and hence topological protection.  Even
if the net is thin (in the $y$-direction) as in Fig.~\ref{fig_thin_net}, order $L/{\ell_{\rm max}}$ isotopy moves are required to move
an excitation around an essential loop and so operate on the ground state manifold.
\begin{figure}[htpb]
\labellist \small\hair 2pt

  \pinlabel $\text{excitation path}$ at 1080 500
  \pinlabel $L\text{-moves}$ at 1195 695

\endlabellist
\centering
\includegraphics[scale=.18]{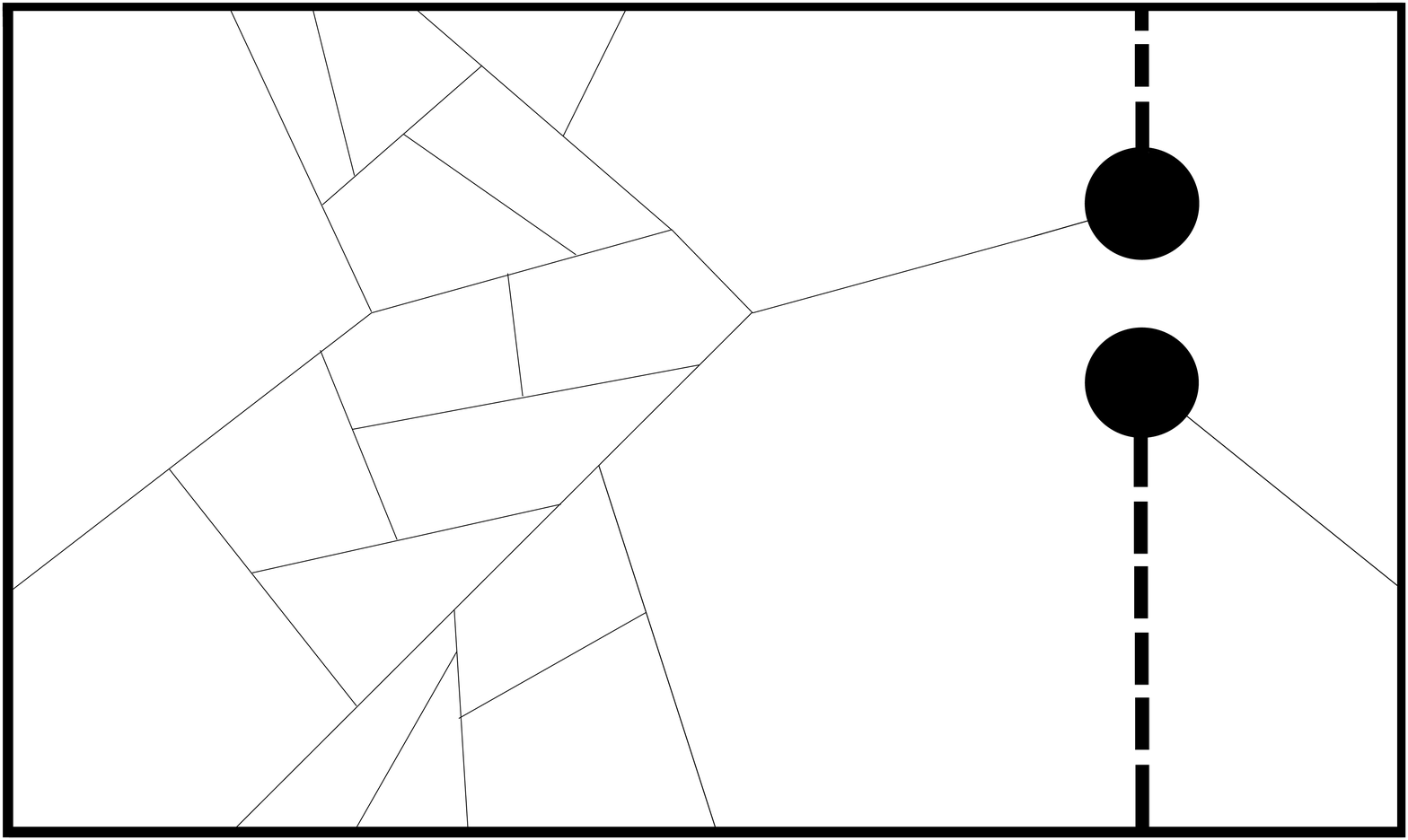}
\caption{A net with a thin part.} \label{fig_thin_net}
\end{figure}

We will now argue that $H_{qg} = H_{qg}^{0+} + \omega B$,
in the pinned context, supports achiral topological phases such as Dfib
and the toric code similar to $H_{LW}$, the Levin-Wen Hamiltonian, yet
coexisting with gapless gravity waves.
% insert A begins
The $+$ superscript in $H_{qg}^{0+}$ indicates that we have added one
more ``between fibers" term to our
basic Hamiltonian $H_{qg}^{0}$, := ``virtual" fusion and ``virtual" $F$-terms. This term raises or lowers the number
of vertices of the net by 2. The Hilbert space is constrained now to
have a maximum of $n$ vertices per net rather than exactly $n$
vertices. The new term introduces (removes) a ``bubble'' into a virtual edge:
$$
\parbox{1.31cm}{\scalebox{0.3}{\includegraphics{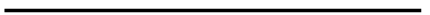}}}
\quad\scalebox{1.3}{$\leadsto$}\quad
\parbox{1.53cm}{\scalebox{0.35}{\includegraphics{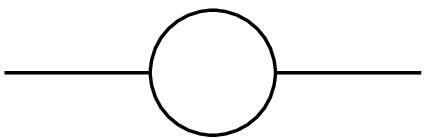}}}
$$
with any allowed effect on labels. Rather than write a general formula
as before, we give as examples the toric code and Dfib cases instead
\begin{eqnarray}
\text{toric code:}\ \ 
\parbox{1.31cm}{\scalebox{0.3}{\includegraphics{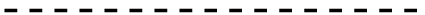}}}
&\scalebox{1.3}{$\leadsto$}&
\frac{1}{\sqrt{2}} \left(
\parbox{1.53cm}{\scalebox{0.35}{\includegraphics{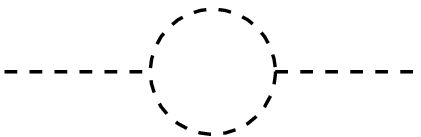}}}
+
\parbox{1.53cm}{\scalebox{0.35}{\includegraphics{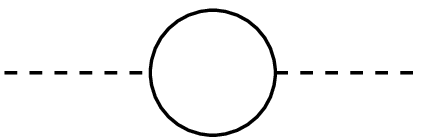}}}
\right) \nonumber \\
\parbox{1.31cm}{\scalebox{0.3}{\includegraphics{line1}}}
&\scalebox{1.3}{$\leadsto$}&
\frac{1}{\sqrt{2}} \left(
\parbox{1.53cm}{\scalebox{0.35}{\includegraphics{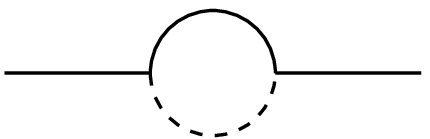}}}
+
\parbox{1.53cm}{\scalebox{0.35}{\includegraphics{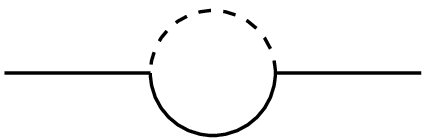}}}
\right) 
\label{fig_nterm1}
\\
\text{Dfib:}\ \ 
\parbox{1.31cm}{\scalebox{0.3}{\includegraphics{line2}}}
&\scalebox{1.3}{$\leadsto$}&
\frac{1}{\sqrt{1+\tau^{-2}}} \left(
\parbox{1.53cm}{\scalebox{0.35}{\includegraphics{bubble1}}}
+ \frac{1}{\tau}
\parbox{1.53cm}{\scalebox{0.35}{\includegraphics{bubble2}}}
\right) \nonumber \\
\parbox{1.31cm}{\scalebox{0.3}{\includegraphics{line1}}}
&\scalebox{1.3}{$\leadsto$}&
\frac{1}{\sqrt{2+\tau^{-1}}} \left(
\parbox{1.53cm}{\scalebox{0.35}{\includegraphics{bubble3}}}
+
\parbox{1.53cm}{\scalebox{0.35}{\includegraphics{bubble4}}}
+ \frac{1}{\sqrt{\tau}}
\parbox{1.53cm}{\scalebox{0.35}{\includegraphics{bubble5}}}
\right) \nonumber
\end{eqnarray}

Just as we did below Fig.~\ref{fig_states}, the relations of
\eqref{fig_nterm1} are easily converted into projector of the form:
$$
 (1-|v\rangle\langle v|) \quad\text{and}\quad (1-|w\rangle\langle w|),
$$
where (in the Dfib case)
\begin{eqnarray}
 v_{unnormalized} &=& 1
   \parbox{1.31cm}{\scalebox{0.3}{\includegraphics{line2}}}
  -\frac{1}{\sqrt{1+\tau^{-2}}}
   \parbox{1.53cm}{\scalebox{0.35}{\includegraphics{bubble1}}}
  -\frac{1}{\sqrt{1+\tau^{-2}}}
   \parbox{1.53cm}{\scalebox{0.35}{\includegraphics{bubble2}}}
 \nonumber \\
 w_{unnormalized} &=& 1
   \parbox{1.31cm}{\scalebox{0.3}{\includegraphics{line1}}}
  -\frac{1}{\sqrt{2+\tau^{-1}}} \left(
   \parbox{1.53cm}{\scalebox{0.35}{\includegraphics{bubble3}}}
%  -\frac{1}{\sqrt{2+\tau^{-1}}}
   +\parbox{1.53cm}{\scalebox{0.35}{\includegraphics{bubble4}}} \right)
  -\frac{1}{\sqrt{2\tau + 1}}
   \parbox{1.53cm}{\scalebox{0.35}{\includegraphics{bubble5}}}
 \nonumber
\end{eqnarray}
These projectors are the new term in $H_{qg}^{0+}$.

Once a bubble has been introduced in a bond, a succession of 6 (or actually
degree of face) many $F$-moves inflates the bubble, carries it around the
face and then collapses it back to its initial state
\begin{figure}[h]
\labellist \small\hair 2pt
  \pinlabel $F$ at 170 350
  \pinlabel $F$ at 368 350
  \pinlabel $F$ at 459 233
  \pinlabel $F$ at 368 118
  \pinlabel $F$ at 170 118
  \pinlabel $F$ at 81 233
\endlabellist
\centering
\includegraphics[scale=0.45]{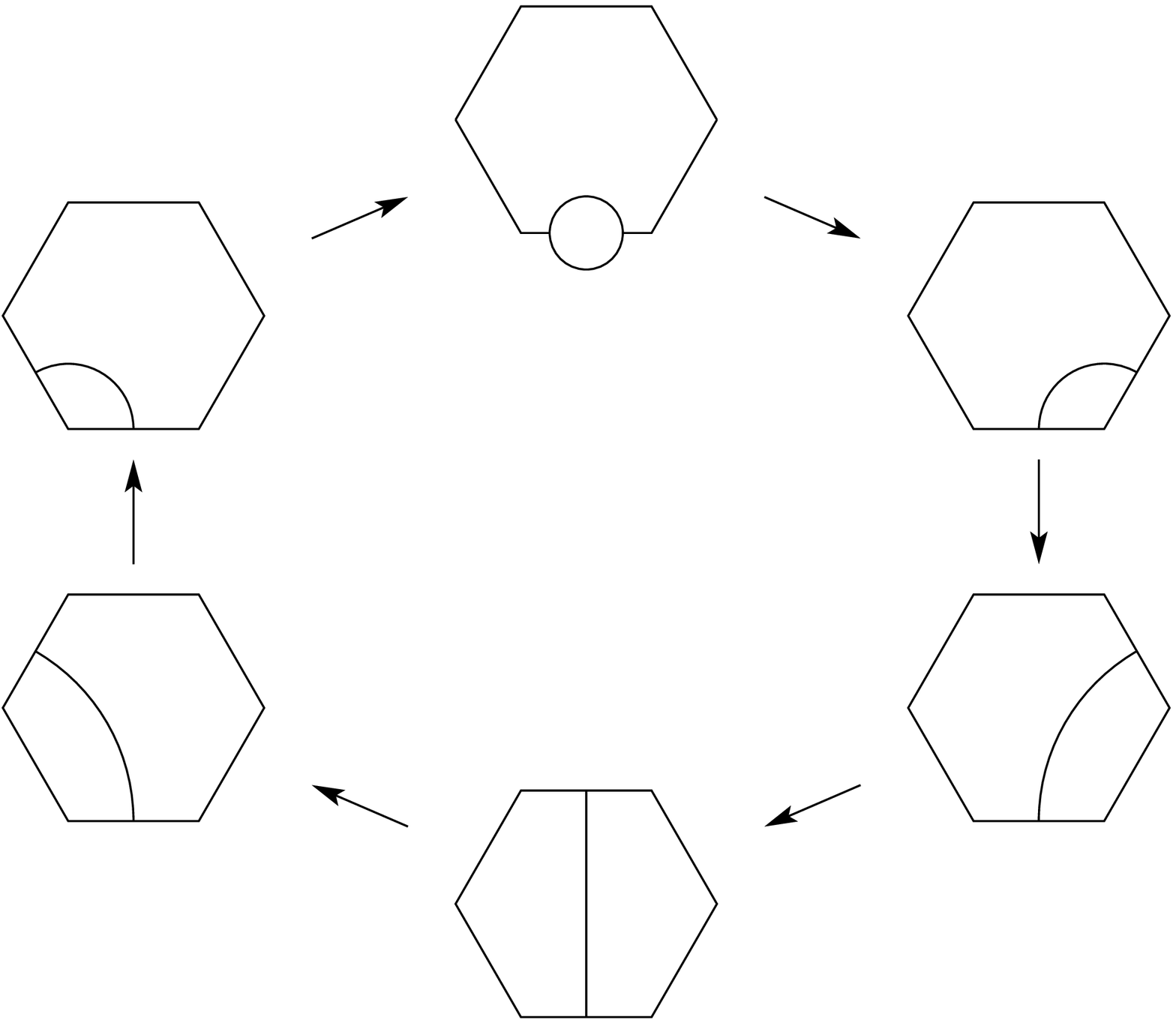}
\label{fig_fmoves}
\end{figure}

This process is easily seen to enforce the LW plaquette condition: the flux through
each plaquette is trivial. The question is how strongly it is enforced. Note that
there is no excitation present in Fig.~\ref{fig_fmoves}, rather we show
a circular
family of ``horizontal'' terms relating one fiber to another, returning
finally to the original fiber. Consequently the cost of a violation of this
emergent ``no flux'' condition is not a high power of small number but rather
proportional to the reciprocal of the number of horizontal terms in the loop.
Again think of a one dimensional ferromagnet with twisted boundary conditions.
The number of horizontal moves is measured by the
nets combinatorics, $6+2$, which is 6 for the trip around the hexagon, 2 for
creation and annihilation of the bubble. Thus despite a suppression by a factor of $\frac{1}{8}$, there is a substantial gap to magnetic excitations in terms of the bare energy
scale of the $F$-move and the ``new'' term in $H_{qg}^{0+}$. Of course the cost of an
electric excitation is precisely the bare energy of the ``vertical'', i.e.
within fiber, terms which enforce the fusion rules.

Although in a random net some plaquettes
will have more than six sides, the probability of $s$ sides decays exponentially
with $s$ (see \ref{Appendix:Numerics}). Thus a small portion ${\mathcal C}$ of 
the configuration space ${\mathcal N}$ with
$s$-gon plaquettes, $s$ large, has cusp-like geometry (as in the case of
hyperbolic geometry) and supports neither small Cheeger cuts nor low lying
eigen functions. It follows that although a magnetic excitation may be cheap
over ${\mathcal C}$ there is no efficient way to tapper off the amplitude
towards zero on ${\mathcal N}-{\mathcal C}$ where magnetic excitations are
expensive. The conclusion is that our analysis for $s\approx 6$ is, in fact,
general and shows a gap to magnetic excitations across all the configuration
space ${\mathcal N}$.

% insert A ends
A detailed comparison of $H_{qg}$ to the exactly solved Levin-Wen Hamiltonian $H_{LW}$ is instructive.  The ground states (in the thermodynamic limit) are expected to be bijective.  The excitations of $H_{qg}$ are, in contrast to $H_{LW}$, mobile.  To build point-like, confined excitations ``wave packets'' will need to be formed.  Combinatorial recoupling arguments show that if such packets are confined in potential wells and braided, the LW (i.e. Jones) braid representation will be exactly realized (in the strong confinement limit).  Thus, we may expect that the entire topological structure, the TQFT, represented by $H_{LW}$ is recaptured by $H_{qg}$.  It is true that braiding will excite gapless gravity waves, but these are visibly non-interacting with the topological information contained in the combinatorics of labeled nets and their recoupling rules.

We would like to explain more fully this remarkable property of the liquid phase. This is the  rigidity of topological information maintained in defiance, so to speak, of the gapless gravity waves which propagate  about.  To do this let us speak metaphorically of the underlying space $X_n$ of (unlabelled) configurations as a ``chain''. This is a reasonable picture since our spectral studies show that the low eigen values of the graph Laplacian have inverse power law scaling similar to the $1/n^{-2}$ scaling of a chain.  We may very roughly view the quasi-geometry of $X_n$ as a string of length $O(n)$. We should worry that very near the ground state energy we will have states whose topological characteristics ``rotate'' as we pass from one end of the chain to the other. Recall our two main exemplars: the toric code and Dfib. Both of these have a $4$ dimensional ground state Hilbert space (torus) spanned by the states $\ket{1}$, $\ket{2}$, $\ket{3}$, and $\ket{4}$. Imagine  a system state that is a family of topological ground states that rotates by $2\pi$ as we move across the length of the ``chain'' $X_n$ and so, on the torus triangulation at ``chain position'' $x$, $0\le x \le L$, we see  the ground state:
$
 \left(\cos(2\pi x/L)\ket{1} + \sin(2\pi x/L)\ket{2}\right).
$
Is such a system state a candidate for a low energy excitation  as $n = L$ approaches infinity? The answer is: ``No''.  To see this look at consecutive ``links'' in the chain, triangulations $\Delta_1$ and $\Delta_2$ with states $\Psi_1$ on $\Delta_1$ and $\Psi_2$ on $\Delta_2$. By the ``code property" of topological ground states (see \cite{MF}), $\Psi_1$ and $\Psi_2$ cannot differ by the application of a local operator. Passing between $\Psi_1$ and $\Psi_2$ will cost energy according to the $H_{qg}^{0+}$ term. In fact, this rigidity is quite robust. Up to the usual caveats about perturbations inducing exponentially fine energy splittings,  it is not possible to deform the ground state as one moves through the configuration space $X$. Since for us the configuration $x \in X$ is a dynamical variable, this is important.  If topological information is stored in this novel phase, when it is retrieved $X$ must  be sampled. $X$ will be sampled according to some distribution and the topological state over the sampled $x$ will then be probed by a quantum measurement.  The output distribution of our probe , when applied to a system ground state,  will be independent of the sampled $x \in X$, as desired.

Because of the pinning and string tension terms, the typical nets in this lattice model are qualitatively similar to the boundaries of Voronoy cells produced by Poisson distributed centers. We recommend this alternative model to the investigation of interested readers.

$H_{qg}$ is not a ``lattice Hamiltonian.''  In particular, it is not
defined on a ``tensor product'' Hilbert space (but rather a fiber-wise
direct sum of these, one for each net in $\mathcal{N}_n$).  Thus, it
is not precise to assert that $H_{qg}$ is ``$k$-body'' for any $k$,
but it is evidently quite simple. One may say that the flux
(plaquette) term of $H_{LW}$, which is 12-body, or more precisely a
6-parameter family of 6-body interactions, has been simulated by the
$F$-move, which in these terms is a 4-parameter family of {\em one}
body interactions.  But to achieve this, we have resorted to a context
where the lattice itself fluctuates and must be counted among the
dynamic variables.

%---------------------------------------------------------------------------------

\subsection{CDT: A (1+1)D Home for Anyons}
Causal dynamical triangulation (CDT) builds layered (1+1)-dimensional ``space-times" by randomly constructing Lorentzian strips as below,
\begin{figure}[htpb]
\labellist \small\hair 2pt

  \pinlabel $\text{$\uparrow$ pseudo-time = $\theta$}$ at 1080 200

\endlabellist
\centering
\includegraphics[scale=.18]{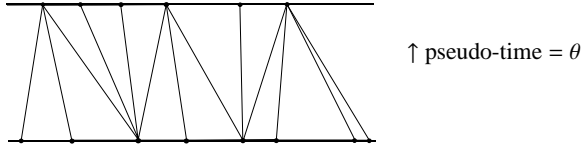}
\caption{One layer of space-time in CDT.} \label{cdt}
\end{figure}
where all horizontal bonds have $\rm{length^2} = 1$, and all other bonds have $\rm{length^2} = -a\le{0}$, where a is a fixed constant, see Fig. \ref{cdt}. The action $S = \int{(\frac{-R}{G}+2\Lambda)}$ is the integrated Regge scalar curvature (appropriate to Lorentz space) plus a suitable cosmological constant. 

We take pseudo-time $\theta$, periodic. It is known \cite{AR} that for suitably chosen $a\approx{.7}$, there is a liquid regime of roughly flat Lorentzian geometries on $\mathbf{S}^1\times{\mathbf{R}}$. This finding offers a remarkable solution in cosmology to the persistent problem of baby universes. We recommend for further study the possibility of importing this innovation into condensed matter physics. The same action can be used to define a density on string nets supported neat flat geometry, and so preserve topological protection in anyonic models based on these geometries.

To visualize the braiding of anyons, described next, picture  $\mathbf{S}^1\times{\mathbf{R}}$ as $\mathbf{R}^2 - \{0\}$ via $(\theta, \rho) \mapsto (e^{\rho},\theta)$. The geometry of $\mathbf{R}^2 - \{0\}$, i.e. its Lorentzian triangulation, is explicitly among the dynamic variable, but in addition the bonds of the triangulation are labeled from a (quantum group) label set, which in this paper is \{1,$\tau$\}. ``Singularities" of the labeling (as explained in detail in Fidkowski {\it et al.}, see \cite{FFNWW}) - annular regions where the state cannot be extended over the disk to a vacuum state - are the ``quasi-particles", or anyons, of Dfib. So a loop of states is a loop of annular Lorentz geometries together with anyons.

One might wonder when importing a 2-dimensional net model from a (1+1)-dimen\-sional quantum gravity model, whether ``causality'' in the model will prevent braiding. If information is not allowed to flow backwards in ``time'' (pseudo-time = $\theta$) we might be unable to braid anyons since they can only move forward in the $\theta$ direction. 
%Because of the 1D radial spacial foliation in $\mathbf{R}^2 - \{0\}$, $F$-moves between dual CDT triangulations can only propel the anyons in a counter-clockwise direction.
 This may appear to limit their possible braidings, but in actuality it does not. A full counter-clockwise 2$\pi$-turn generates the center $C$ of each braid group, $B_n, n\ge{3}$. Thus as the anyons move radially back and forth, their overall progression in the pseudo-time = $\theta$ direction only multiplies the braid by a central element - corresponding in each irreducible sector of the Jones representation to an irrelevant overall phase. Consequently, the causality of the construction (in pseudo-time = $\theta$) does not restrict the image of the braid representations and building topological phases is not hampered by a causality constraint.
\begin{figure}[htpb]
\labellist \small\hair 2pt
  \pinlabel $\text{$\leftarrow t = \rm{time}$}$ at 362 142
  \pinlabel $\text{$\theta=\rm{pseudotime}$}$ at 410 30
  \pinlabel $\text{$\leftarrow$}$ at 300 30
\endlabellist
\centering
\includegraphics[scale=.3]{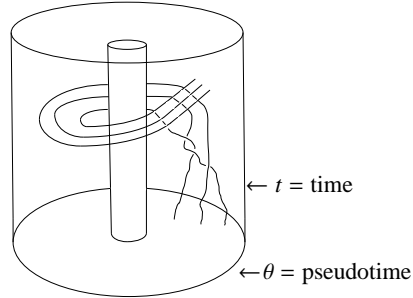}
\caption{Anyon trajectories superimposed on a central twist.} \label{fig_twist}
\end{figure}

\section{Conclusions}
\label{sec:conclusions}

Because nuclei are heavy lattices in condensed matter are generally thought of as fixed or classical degrees of freedom. It is true that in chemistry superpositions of isomers can be important but generally  the lattice is not takes to be a dynamical quantum mechanical variable. In this paper this is exactly what we do. We have not forgotten that nuclei are heavy; we imagine that there may be models in which some electron degrees of freedom define a lattice and others decorate it and that both should be allowed to fluctuate. This paper is not about a specific model of this kind but rather a preliminary survey of the hazards and prospects that await us ``off lattice''. Our focus has been entirely on building topological phases, although off lattice models may have wider applicability.

We have seen that the chief hazard is uncontrolled fluctuations in the now random geometry of the lattice, a phenomenon colorfully called ``baby universe'' in the quantum gravity literature. These fluctuations threaten to destroy the ratio of scales $L/a$, system length / lattice spacing, on which topological protection, ${\rm error} \sim \exp(-{\rm const.}L/a)$ depends. We have also seen our three approaches (methods of birth control) to retaining topological protection. Briefly they were: 1) a minimally fluctuating {\em crystalline phase}, unsatisfactory due to a vanishingly small excitation gap $\Delta$ for the topological phase. 2) Pinning the fluctuating lattice to a background (lattice or continuum). This seems to work but sacrifices some of the simplicity we hoped to find in off lattice models. 3) Causal Dynamical Triangulations (CDT). Here we borrow the solution (as well as the problem) from the quantum gravity community.

In order to evaluate the impact of geometry fluctuations, analytical (\ref{Appendix:C}) and numerical (\ref{Appendix:Numerics}) work was done on the statistics of loop gases and string nets. In \ref{Appendix:Spectrum} and \ref{Appendix:OuterPlanar} Cheeger's theory relating the geometry of manifolds to their vibrational modes is adapted to the infinite dimensional context to construct estimates of the spectrum of Hamiltonians derived from our models. In particular we develop a method for constructing upper bounds to the spectral gap of a Hamiltonian $H$ by Monte Carlo studies (of both the gap and estimates for Cheeger's constant) on a weighted graph $G$ derived from $H$. We find (method 2 and appendices) regimes in which topological information is protected while coexisting with gapless ``vibrational'' modes across the space of geometries. Such results encourage us to regard the off-lattice approach as viable and worthy of continued investigation.

%%%%%%%%%%%%%%%%%%%%%%%%%%%%%%%%%%%%%%%%%%%%%%%%%%%%%%%%%%%%%%%%%%%%%%%%%%%

\section*{Acknowledgments}
We thank the Aspen Center for Physics for hospitality. This work was supported, in part, by the Swiss National Science Foundation and 
the Swiss HP$^2$C initiative.

%%%%%%%%%%%%%%%%%%%%%%%%%%%%%%%%%%%%%%%%%%%%%%%%%%%%%%%%%%%%%%%%%%%%%%%%%%%

\appendix
\renewcommand{\thetheorem}{\Alph{section}.\arabic{theorem}}

\section{Spectrum of Graph Laplacians and Other Local Hamiltonians}
\label{Appendix:Spectrum}

We start with the following data: a finite dimensional Hilbert space
$\mcH$ spanned by a preferred set of basis kets $\{\ket{i}\}$, a
Hamiltonian $H:\mcH\rightarrow\mcH$, and a known ground state wave
function $\psi_0 = \sum_{i}a_i\ket{i}$ for $H$. We construct from the
data a weighted graph $G$ whose Laplacian
$\mathcal{L}:\mcH\rightarrow\mcH$ is {\em easy} to study
numerically. 
We focus on the first eigenvalue $\lambda_1(\mathcal{L})$ and 
 verify the Cheeger inequalities for $G$: 
\begin{equation}
  \label{eq:cheeger_inequality}
  2h_G \geq \lambda_1(\mathcal{L}) \geq \frac{h_G^2}{2} \,.
\end{equation}
If it should happen that $\lambda_1(\mathcal{L})\rightarrow 0$ (as
a scaling limit is taken), we may also conclude that the original
$H$ is gapless (in the same limit.) This is because a small $h_G$
means a {\em neck} in the set of kets $\{\ket{i}\}$ with little
coupling from $cH$ (where $c$ is a positive running factor, perhaps
proportional to $(\textit{system size})^{-1}$, arising in the proof)
from left to right sides of the neck. The trial wave of the form
$\psi_1 = b_1\psi_0^{\text{left}}-b_2\psi_0^{\text{right}}
(b_1,b_2>0)$ will be orthogonal to $\psi_0$ and satisfy
$c(\bra{\psi_1}H\ket{\psi_1}-\bra{\psi_0}H\ket{\psi_0})\rightarrow
0$. If this rate of convergence to zeros is faster than $c$, it will
imply $H$ gapless. Conversely, if we know a quantum mechanical system
is gapped (e.g. the Levin-Wen model \cite{LW}), it will imply a tightly
connected geometry for the appropriate weighted graph of string net
configurations. 

The reader may wonder what good is a method for studyinig the ``gap''
if it requires knowledge of the ground state $\psi_0$. In the  case of
a topological phase, one may begin with a formula for the ground state
wave function (given by $d$-isotopy, or a chromatic evaluation) and
from this, attempt to build a gapped Hamiltonian. This appendix
provides ammunition for shooting down such Hamiltonians (i.e. showing
them gapless) as in  \cite{FNS05}. 

Here is the construction. The vertices of $G$ are simply the index set
$\{i\}$ for the kets of $H$. We set the edge weight $w_{ij} =
c|H_{ij}|$ provided $i\neq j$, for a positive constant $c$ yet to be
determined. The {\em wave function} weight at  $i$ is
$d_i=|a_i|^2$. Write $d_i = c\sum_{j\neq i}|H_{ij}|+w_{ii}$ where
$c$ is the largest ($i$ independent) constant allowing all $w_{ii}\geq
0$. This fixes $c$ and the $w_{ii}$. The $cw_{ii}$ are to be thought
of as weights on loops at $i$. Thus, $G$ has edge weights $cw_{ij}$
and vertex weights $d_i$. Following Chung \cite{Chung}, there are
``unweighted'' and ``weighted'' operators $L$ and $\mathcal{L}$, both 
of which are symmetric and have a
zero mode:
\begin{eqnarray}
  \label{eq:lap_def_1}
  L_{ij} &=&
  \begin{cases}
    d_i-w_{ii} & \text{if }i=j\\
    -w_{ij} & \text{otherwise}
  \end{cases}\\
  \mathcal{L} &=& T^{-1/2}LT^{-1/2},
\end{eqnarray}
where $T$ is the diagonal matrix with $T_{ii} = d_i$. Explicitly,
\begin{equation}
  \label{eq:lap_def_2}
  \mathcal{L} =
  \begin{cases}
    1-\frac{w_{ii}}{d_i} & \text{if }i=j\\
    \frac{-w_{ij}}{\sqrt{d_id_j}} & \text{otherwise} \,.
  \end{cases}
\end{equation}

\begin{example}
  As a sanity check on the method, we check explicitly, in a toy model,
  that the choice of basis only affects the spectrum
  slightly (by a factor of 2.) We explore in the simplest case the
  dependence of the spectrum of $\mathcal{L}$ on the choice of bases
  for $\mcH$. Let $\mcH = \C^2$ and
  \begin{displaymath}
    H_{\theta} =
    \begin{pmatrix}
      \cos\theta & -\sin\theta\\
      -\sin\theta & -\cos\theta
    \end{pmatrix}
  \end{displaymath}
  $H_{\theta}$ annihilates $
  \begin{pmatrix}
    \sin{\theta}/{2}\\\cos{\theta}/{2}
  \end{pmatrix},
  $ so we need to solve, with $c$ as large as possible, the equations
  \begin{eqnarray*}
    d_1 &=& \sin^2\frac{\theta}{2} = c |\sin\theta|+w_{11}\\
    d_2&=& \cos^2\frac{\theta}{2} = c |\sin\theta|+w_{22}
  \end{eqnarray*}
  Recalling that $\sin\theta =
  2\cos\frac{\theta}{2}\sin\frac{\theta}{2},$ the natural (and
  correct) guess in the interval $0\le\theta\le\pi/2$ is $c = \frac{1}{2}\tan\frac{\theta}{2}$. This yields: 
  \begin{displaymath}
    w_{11} = 0\hspace{1cm}\text{and}\hspace{1cm}w_{22} = \cos\theta
  \end{displaymath}
  Substituting, we find:
  \begin{displaymath}
    \mathcal{L} =
    \begin{pmatrix}
      1 & \tan\frac{\theta}{2}\\
      \tan\frac{\theta}{2} & \tan^2\frac{\theta}{2}
    \end{pmatrix}
  \end{displaymath}
  Solving for the eigenvalues we obtain $\lambda_1=1/\cos^2\frac{\theta}{2}$, which in the considered interval $0\le\theta\le\pi/2$ varies only between 1 and 2. In the other intervals we get similar results.

\end{example}

With this small check of quasi-invariance of $\lambda_1$ under basis change,
we derive the Cheeger inequalities in the relevant weighted graph
setting, closely following \cite{Chung}. 

$\mathcal{L}$ acts on functions of $G$ by (left) multiplication.  The
lowest eigenvalue is $\lambda_0 = 0$ with eigenfunction $f_0(i) =
\sqrt{d_i}$.  When $G$ is connected, $\lambda_0$ is non-degenerate.
We will be concerned with the next eigenvalue $\displaystyle \lambda_1
= \inf_f \frac{\langle f \mathcal{L} f \rangle}{\| f \|_2}$ computed
with respect to $\mu$, the measure or vertex weighting with
weight($i$) = $d_i$, for $f$ orthogonal to constants.  We denote
$\lambda_1$ by $\lambda_G$ or just $\lambda$ and use $f$ for its
eigenfunction. 

We define the Cheeger constant $h$: $$h = \min_{S \subset V}
\frac{F_w(S, \bbar{S})}{\min(Vol(S), Vol(\bbar{S}))} \,,$$ where $S$
is an arbitrary subset of the vertex set $V(G)$, $\bbar{S}$ is $V
\setminus S$, $F_w$ denotes the weight of edges between $S$ and
$\bbar{S}$, $\displaystyle F_w = \sum_{i \in S, j \in \bbar{S}}
w_{ij}$.  Finally, $\displaystyle Vol(S) = \sum_{i \in S} d_i$,
$\displaystyle Vol(\bbar{S}) = \sum_{j \in \bbar{S}} d_j$. 

\begin{theorem}
\label{theorem_A1}
$2h \geq \lambda$
\end{theorem}

\begin{proof}
  Let $S$ achieve $h$ and set $a=Vol(S)$ and $b=Vol(\bbar{S})$. Define a
  ``trial'' eigenfunction: $$f_i = \left\{ \begin{array}{lc}
      \frac{1}{a}, & i \in S \\ -\frac{1}{b}, & i \in
      \bbar{S} \end{array}\right.$$ We have, from the Rayleigh-Dirichlet
  integral: 
  \begin{align*}
    \lambda &\leq F(S, \bbar{S})\frac{(\frac{1}{a} + \frac{1}{b})^2}{\frac{1}{a^2}a + \frac{1}{b^2}b} \\ &= F(S, \bbar{S}) \left( \frac{1}{a} + \frac{1}{b} \right) \\ &\leq \frac{2F(S,\bbar{S})}{\min(Vol(S),Vol(\bbar{S}))} \\ &= 2h
  \end{align*}
\end{proof}

\begin{theorem}
  \label{theorem_A2}
  $\lambda \geq \frac{h^2}{2}$.
\end{theorem}

\begin{proof}
  Let functions $f,k:V(G)\rightarrow\R$ be related by $k_i =  d_if_i.$
  Now the Rayleigh quotient:  
  \begin{eqnarray*}
    \faktor{\innerprod{k}{\mathcal{L}k}}{\innerprod{k}{k}} &=& \faktor{\innerprod{k}{T^{-1/2}LT^{-1/2}k}}{\innerprod{k}{k}}\\
    &=& \faktor{\innerprod{f}{Lf}}{\innerprod{T^{1/2}f}{T^{1/2}f}}\\
    &=& \faktor{\sum_{i\sim j}|f_i-f_j|^2}{\sum_if_i^2d_i}
  \end{eqnarray*}
  becomes $\lambda$ when minimized among $k_i$ orthogonal to $d_i$,
  equivalently by $f$, orthogonal to constants. We assume $f$ is such a
  minimum. Thus $\mathcal{L}k=\lambda k$. 
  
  We index the vertices $i$ of $G$ in $f$-increasing order, $f_i \leq
  f_{i+1}$, and without loss of generality assume $\displaystyle
  \sum_{f_i < 0} d_i \geq \sum_{f_j > 0} d_j$.  For each $i \in V$ let
  $\displaystyle c_i = \sum_{j \leq i < k} w_{jk}$, measures the
  $i$th ``cut'' between $S_i = \{j \leq i\}$ and $\bbar{S}$. Set
  $\displaystyle \beta = \min_{i \in V}
  \frac{c_i}{\min(Vol(S),Vol(\bbar(S)))}$.  Clearly $\beta \geq k$. 
  
  We set $V_+ = \{i | f_i \geq 0\}$ and $E_+$ the set of edges with at
  least one endpoint in $V_+$.  Finally, set $\displaystyle g_i =
  \left\{ \begin{array}{lc} f_i & \text{iff } i \in V_+ \\ 0 &
      \text{otherwise} \end{array}. \right.$ 
  Now compute $\lambda$:
  \begin{align*}
    \lambda &= \sum_{i \in V_+} \frac{ f_i \left( \sum_{(i,j) \in E_+} w_{ij}(f_i - f_j)\right)}{\sum_{i \in V_+} f_i^2 w_i}\\
    &\hspace{1cm}\text{(cutting off some numerator terms)} \\
    &\geq \frac{\sum_{(i,j) \in E_+} w_{ij}(g_i - g_j)^2}{\sum_{i \in V} g_i^2 w_i}\\
    &= \frac{\left( \sum_{(i,j) \in E} w_{ij}(g_i - g_j)^2 \right) \left( \sum_{(i,j) \in E} w_{ij}(g_i + g_j)^2 \right)}{\left( \sum_{i \in V} g_i^2 w_i \right) \left( \sum_{(i,j) \in E} w_{ij}(g_i + g_j)^2 \right)}\\
    &\hspace{1cm}\text{(by Cauchy-Schwartz)} \\
    &\geq \frac{\left( \sum_{(i,j) \in E} w_{ij}(g_i^2 - g_j^2) \right)^2}{\left( \sum_{i \in V} g_i^2 w_i \right) \left( \sum_{(i,j) \in E} w_{ij}(g_i + g_j)^2 \right)}\\
    &\hspace{1cm}(\text{Since }\sum_E w_{ij}(g_i + g_j)^2 \leq 2\sum_V g_i^2 w_i) \\ 
    &\geq \frac{\left( \sum_{(i,j) \in E} w_{ij}(g_i^2 - g_j^2) \right)^2}{2 \left(\sum_V g_i^2 w_i \right)^2}\\
    &\hspace{1cm}\text{(discarding cross terms from the numerator)} \\
    &\geq \frac{\left( \sum_i c_i |g_i^2 - g_j^2| \right)^2}{2 \left(\sum_V g_i^2 w_i \right)^2} \geq \frac{\left( \sum_i \beta Vol(S_i) |g_i^2 - g_j^2| \right)^2}{2 \left(\sum_V g_i^2 w_i \right)^2}\\
    &\hspace{1cm}\text{(telescoping the sum)}\\
    &= \frac{\left( \sum_i \beta w_i g_i^2 \right)^2}{2 \left(\sum_V g_i^2 w_i \right)^2} = \frac{\beta^2}{2} \geq \frac{h^2}{2}
  \end{align*}
\end{proof}

To better understand the proof of Theorem \ref{theorem_A2}, we
summarize Cheeger's original argument in the context of a Riemannian
manifold $M$.  Define Cheeger's constant $h$ by: $$h = \inf_{S \text{
    separating } M} \frac{Area(S)}{Volume(M)}$$ Let $f$ be the first
eigenfunction of the Laplacian $\Delta$ orthogonal to constants. 
\begin{align*}
\lambda &= \frac{\int f \Delta f}{\int f^2} = \frac{\int f \Delta f}{\int f^2} \frac{\int f^2}{\int f^2} \geq \frac{\left( \int |f||\nabla f| \right)^2}{\left( \int f^2 \right)^2} \\ 
&\geq \frac{1}{4} \frac{\left(\int (\nabla f^2)\right)^2}{(\int f^2)^2}
\end{align*}

Define $t = f^2$ as a parameter on $M$ and apply the co-area formula
to the $t$-levels to obtain: 
\begin{align*}
\int (\nabla f^2) &= \int Area(t \text{-level}) dt \leq h \int Vol[0,t] dt \\
&= -h \int t \frac{d Vol}{dt}dt = -h \int t d Vol = -h \int f^2 d Vol.
\end{align*}

Thus, $$\lambda \geq \frac{1}{4} \frac{\left( -h \int f^2 d Vol
  \right)^2}{\left( \int f^2 d Vol \right)^2} = \frac{h^2}{4}$$ 
\hfill \ensuremath{\Box}

 We offer a protocol which {\em may} succeed in verifying that a
 quantum mechanical Hamiltonian $H_n:\mcH_n\rightarrow\mcH_n$ is
 {\em gappless} above its (known) ground state $(\psi_0)_n$ as a
 limit $n\rightarrow\infty$ is taken. 
\begin{protocol}
  Select preferred kets $\ket{i}$ for $\mcH_n$ (we do not clutter the
  notation by showing the dependence of the index set $\{i\}$ on $n$.)
  use $H_n, (\psi_0)_n$ to construct the weighted graph $G_n$ as
  above, and set $\lambda^n = \lambda_1(\mathcal{L}_{G_n})$. Recall
  that the construction of $G_n$ requires extracting a constant $c_n$
  (in our two dimensional example $c^n = c^{\theta} =
  (2\cos\frac{\theta}{2})^{-1}$,) the minimal suppression factor for
  interactions $H_{ij}$ required to normalize the vertex weights $d_i := |a_i|^2$
  be positive and with $\sum_id_i = 1$. Compute the ratio
  $\frac{D_n\sqrt{\lambda^n}}{c_n},$ where $D_n = \max_i|a_{i,n}|^2$
  for $\psi_{0,n} = \sum_ia_{i,n}\ket{i}$. 
  \begin{claim}
    If $\frac{D_n\sqrt{\lambda^n}}{c_n}\rightarrow 0$, then $H_n$ is gapless, i.e. $\lambda_1(H_n)-\lambda_0(H_n)\rightarrow 0$ as $n\rightarrow\infty$.
  \end{claim}
  \begin{proof}
    We have checked $\lambda^n = \lambda_1(\mathcal{L}_{G_n})\geq
    \frac{h^2_{G_n}}{2}$, so $h_n = h_{G_n}<\sqrt{2\lambda^n}$. As in
    the proof of (\ref{theorem_A1}), let $S^n$ achieve $h_n$ and
    define: 
    \begin{align*}
      \psi_1^n(\ket{i}) &= \frac{1}{a_n}(\frac{1}{a_n}+\frac{1}{b_n}),&i\in S^n, a_n = Vol(S^n)\\
      \til{\psi_1^n}(\ket{i}) &= -\frac{1}{b_n}(\frac{1}{a_n}+\frac{1}{b_n}),&i\in \bbar{S}^n, b_n = Vol(\bbar{S}^n).
    \end{align*}
    Then, suppressing the $n$ super/subscripts,
    \begin{align*}
      \bra{\psi_1}H\ket{\psi_1} - \bra{\psi_0}H\ket{\psi_0} & = E_1-E_0\\
      &\leq \sum_{\substack{i\in S\\j\in S}}\frac{1}{a}\frac{1}{1-a}|a_i||a_j||H_{ij}|\\
      &\leq \frac{1}{a}\frac{1}{1-a}D~F_{w}(S,\bbar{S})c^{-1}\\
      &=: \frac{1}{a}\frac{1}{1-a}DFc^{-1},
    \end{align*}
    where $a=Vol(S)$ and without loss of generality, $Vol(S)\leq
    Vol(\bbar{S})$. From Cheeger's inequality (\ref{theorem_A2}), we
    have: 
    \begin{displaymath}
      \frac{F}{a} = h\leq \sqrt{2\lambda},
    \end{displaymath}
    so
    \begin{displaymath}
      E_1-E_0 \leq \frac{1}{a}DFc^{-1}\leq D\sqrt{2\lambda}c^{-1}.
    \end{displaymath}
\end{proof}

This protocol allows a systematic approach for vetting models which
produce known topological wave functions as the ground state (say on a
2-sphere) but may not be gapped above the ground state. There have
been previous successes in showing models gapless by finding directly
the Cheeger cut into $V(G) = S\amalg\bbar{S}$ \cite{FNS}. The present protocol
may be more practical as less geometric insight is required, unfortunately plugging in the analytical bounds from Appendix C  into the claim we find $\frac{D_n\sqrt{\lambda^n}}{c^n}\approx{n^{{1}/{4}}}$. Using the numerical scalings in  Appendix D we find $\frac{D_n\sqrt{\lambda^n}}{c^n}\approx{n^{{1}/{8}}}$, which does not approach zero as $n\rightarrow\infty$ either. Further geometric insight into the graph $G_n$ might allow one to use $\lambda^n$ rather than $\sqrt{\lambda^n}$ in the claim, yielding $n^{{-3}/{4}}$. This would be legitimate if $G_n$ looked spectrally more like a tree than a line.

The $d=1$ loop gas \cite{FNS} has the surprising feature that the very
same ground state arises as a gapped and gapless ground states of two
different Hamiltonians \cite{FNS05}.  A second example
($d = \sqrt{2}$) was proven \cite{TTSN} via decay of spatial
correlators never to arise as a gapped ground state for any local
Hamiltonian.  However, in many cases unlike the above, one will not be
so fortunate to find a narrow cut for $G$.  Rather, more generically
one may expect to learn something about the spectrum ($\lambda$) of
$\mathcal{L}$ on $G$ and perhaps some properties of the first
eigenfunction $f$ via Monte Carlo methods applied to $G$ (since this
problem is completely classical).  In this case, one should try to use
the protocol. In a gapless system, to find the precise
power at which $\lambda^n\rightarrow 0$, more refined trial wave
functions involving a gradual, not abrupt, phase change across the cut
should be studied, as in \cite{FNS}. 

\begin{note}\label{note_2}
Our protocol can be used in
contrapositive form to argue that the $\lambda$ associated to certain
weighted graphs of configurations cannot decay too quickly in system
size when we know that the (weighted) graph arises as the ground state
of a gapped Hamiltonian, such as the Levin-Wen model.  Specifically,
when applied to the Fibonacci anyons, one may argue that the set of
subgraphs $G$ of the honeycomb when weighted by its topological
evaluation (see \cite{FFNWW}) and supplied with edges corresponding to
a bounded number of $F$-moves and local circle creation/deletions on a
finer scale honeycomb must have $\lambda_{G}$ decaying no faster than
$(\text{system size})^{-1}$. 
\end{note}
\end{protocol}

Monte Carlo methods may eventually be able to extract some information
of the eigenfunction $f$ associated to
$\lambda_1(\mathcal{L}_{G^n})$. It is reasonable to suppose that
knowledge of $f$ could refine the previous protocol. The final
paragraphs of \ref{Appendix:Spectrum} present, schematically,
the outlines of a complementary approach to extracting information on
the quantum mechanical spectrum ($H$) from the first eigenfunction of
the classical $\mathcal{L}_G$. Let us use $f$ to build a trial wave
function $\Psi_1 = f\Psi_0$ from the ground state $\Psi_0$ of $H$.
\begin{equation}\label{eqn_1_thmA3}
\lambda_{1,H}-\lambda_{0,H} \leq \langle \psi_1 |H| \psi_1 \rangle - \langle \psi_0  |H| \psi_0 \rangle = \sum_\alpha \langle f \psi_0 | T_\alpha | f \psi_0 \rangle - \sum_\alpha \langle \psi_0 |T_\alpha| \psi_0 \rangle
\end{equation}
where we have written $H = \sum_{\alpha}T_{\alpha}$ as a sum of local
terms. For any term $T_\alpha$ which acts at state $i \in V(G)$ we
should study the variation of the quadratic forms on the right hands
side of equation \ref{eqn_1_thmA3} at second order in the gradient
$\nabla f$.  (The 0-th order variation vanishes since $f$ is
normalized, $\int f^2 d\mu - \int 1^2 d\mu = 0, \sum_i |a_i|^2 = 1$.
After summing over $\alpha$, first order variation must also vanish
since $\langle \psi_0 |H| \psi_0 \rangle$ is critical for (actually
minimizes) expectation.) 

At $2$nd order in $\nabla f$ and with $\alpha$ fixed,
\begin{equation}\label{eqn_2_thmA3}
\text{r.h.s.}_\alpha(\text{\ref{eqn_1_thmA3}})\hspace{0.25cm}\begin{matrix}<\vspace{-0.5cm}\\\sim\end{matrix}\hspace{0.25cm}
\lambda_{T_\alpha}^{\max} \| \nabla f_i \|^2
\end{equation}
where $\lambda_{T_{\alpha}}^{\max}$ measures the largest eigenvalue of
$T_{\alpha}$ after normalizing all eigenvalues to be positive. In
(\ref{eqn_2_thmA3}), $i$ ranges over states on which $T_\alpha$
operates.  If $n$ is the maximum number of terms $T_\alpha$ operating
on any state $|i\rangle$, we may ``integrate over $\alpha$'' to obtain
from (\ref{eqn_2_thmA3}): 

\begin{align*}
\text{r.h.s.(\ref{eqn_1_thmA3})} &\leq n \lambda_T^{\max} \sum_{i \in V(G)} \| \nabla f_i \|^2 \\ 
&\leq n \lambda_T^{\max} \lambda_G,\hspace{1.5cm}\text{for }\lambda_{Y}^{\max} = \max\{\lambda_{T_{\alpha}}^{\max}\}.
\end{align*}
So, at least schematically, there should be an estimate:
\begin{displaymath}
  \lambda_{1,H}-\lambda_{0,H}\leq n\lambda_{T}^{\max}\lambda_G.
\end{displaymath}

We finally wish to mention a related paper \cite{Iannis} which uses similar methods to argue for the existence of gapped models.
%\[\begindc{0}[3]
%    \obj(5,30){$\text{basis}(H_N)$}
%    \obj(35,30){$\text{basis}(\bbar{H})$}
%    \obj(35,10)[A]{$\text{\{Nets\}}$}
%    \mor(10,30)(30,30){}
%    \mor(35,30)(35,10){}
%\enddc\]
%
%\begin{figure}[htpb]
%\labellist \small\hair 2pt
%  \pinlabel $\phi($ at 60 257
%  \pinlabel $)$ at 147 257
%  \pinlabel $=$ at 190 257
%  \pinlabel $\displaystyle\sum_mF_{kln}^{ijm}\phi($ at 233 254
%  \pinlabel $)$ at 333 257
%\endlabellist
%\centering
%\includegraphics[scale=1]{pentreln2}
%\caption{} \label{pentreln}
%\end{figure}
%

%------------------------------------------------------------------------------

\section{Outer Planar Triangulation}
\label{Appendix:OuterPlanar}

%\doublespacing

An outer planar triangulation ($OPT$) is a triangulation of the $n$-gon $P_n$ in which no new vertices in the interior disk are permitted.  The $n$-gon is given a fixed base point vertex and orientation.  Thus, for $n =$ 3, 4, 5, 6,$\dots$ the number of $OPT$ are 1, 2, 5, 14,$\dots$.  In general, $|OPT_{n+2}| = c_n = \frac{1}{n+1} {2n\choose n}$, the $n$th Catalan number.  This statement may be familiar as the correct counting of dual planar trivalent trees.

Let $G_n$ be the abstract graph with vertices $OPT_n$ and edges determined by ``diagonal flips'' defined on quadrilaterals made from a pair of triangle sharing a bond.  Give all vertices and edges of $G_{n+2}$ unit weight.  The spectrum $L$ of $G$ may be similar to that of the more interesting case of triangulations of the 2-sphere and because there are simple asymptotic formulas for the Catalan number, we can explicitly compute a lower bound $k_n \leq O(n^{-1/2})$ for our Cheeger-like isoperimetric constant $\displaystyle k_n = \min_{S \subset V(G_n)} \frac{E(S, \bbar{S})}{\min(Vol(S), Vol(\bbar{S}))}$ where $E(S,\bbar{S})$ counts edges from $S$ to $\bbar{S} = V(G_n) \setminus S$ and $Vol(S) = \#$ vertices in $S$.  We thank Oded Schramm for guiding us through this example.

For simplicity (only) take $n$ odd.  Now there will certainly be a unique ``central triangle'' $\Delta$ with the property the three connected bits of sides($P_n$) all contain less than $\frac{n}{2}$ sides.  Call the ``lengths'' of these three bits $\frac{n}{2} > n_1 \geq n_2 \geq n_3$, $n_1 + n_2 + n_3 = n$.  We divide $OPT_n$ into two disjoint pieces, $thick_n \cup thin_n = OPT_n$, and $thick_n \cap thin_n = \emptyset$ according to whether $n_3 \geq \frac{n}{10}$ (called $thick_n$) or $n_3 < \frac{n}{10}$ (called $thin_n$).  

We will use the well known relation $c_{n \pm const} = O(1) 4^n n^{-\frac{3}{2}}$ and in the future use $\approx$ to absorb the $O(1)$.

\begin{align*}
|thick_n| &\approx \sum_{\frac{n}{2} > n_1 \geq n_2 \geq n_3 \geq \frac{n}{10}} 4^{n_1} n_1^{-\frac{3}{2}} 4^{n_2} n_2^{-\frac{3}{2}} 4^{n_3} n_3^{-\frac{3}{2}} \\
&= \sum_{\frac{n}{2} > n_1 \geq n_2 \geq n_3 \geq \frac{n}{10}} 4^n (n_1 n_2 n_3)^{-\frac{3}{2}} \\
&\approx \sum_{O(n^2) \text{ terms}} 4^n (n^3)^{-\frac{3}{2}} \\
&\approx 4^n n^{-\frac{5}{2}}
\end{align*}

and

\begin{align*}
|thin_n| &\approx \sum_{1 \leq s \leq \frac{n}{10}} (\# (n_1,n_2) \text{ with } n_3 = s)(\# (n_1,n_2,n_3 = s)) \text{ configurations} \\
&\approx \sum_{1 \leq s \leq \frac{n}{10}} s \left(4^{n_1} n_1^{-\frac{3}{2}} 4^{n_2} n_2^{-\frac{3}{2}} 4^{n_3} s^{-\frac{3}{2}}\right), \ \text{for typical}\ n_1,n_2\ \text{with}\ n_1+n_2=n-s \\
&\approx 4^n \sum_{1 \leq s \leq \frac{n}{10}} s^{-\frac{1}{2}} n^{-3} \\
&\approx 4^n n^{\frac{1}{2}} n^{-3} \\
&= 4^n n^{-\frac{5}{2}}
\end{align*}

So both $thick$ and $thin$ portions of $OPT_n$ have $O(1)$ proportion of all the vertices on $G_n$.  

Now consider the probability of being ``near,'' within $a$ of the boundary between thick and thin: $|n_3 - \frac{n}{10}| < a$.  Call such configurations ``boundary'' or $\partial_a$, $|\partial_a| \approx a 4^n n^{-\frac{7}{2}}$ as there are $O(na)$ such numerical configurations each occurring order $4^n n^{-\frac{9}{2}}$ ways.  Only diagonal flips on one of the three sides of the central triangle can possibly affect membership in $thick_n$ and $thin_n$, and we should estimate how many such flips can relate $thin$ to $thick$.  The largest contribution comes from flips on side $n_1$ (or equivalently $n_2$) in which $a$ vertices of the $n$-gon move to $n_3$ where $a = o(n)$.  We estimate the number of such $G$-edges as follows:

\begin{figure}[htpb]
\labellist \small\hair 2pt
%  \pinlabel $n_1$ at 70 70
%  \pinlabel $n_2$ at 103 80
%  \pinlabel $n_3$ at 102 33
%  \pinlabel $b$ at 55 33 
  \pinlabel $n_1$ at 32 50
  \pinlabel $n_2$ at 66 58
  \pinlabel $n_3$ at 62 8
  \pinlabel $b$ at 18 8
\endlabellist
\centering
\includegraphics[scale=1]{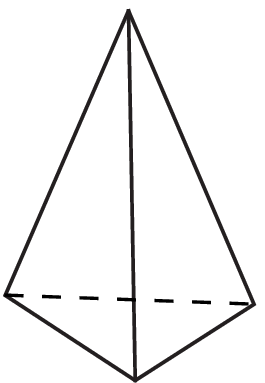}
\caption{} \label{fig_outerplanar}
\end{figure}

%\begin{align*}
%E_1(thin_n, thick_n) &\approx \sum_{a \leq o(n)} 2^n a n^{-\frac{7}{2}} \sum_{\frac{n_1}{2} \geq b \geq a} b^{-\frac{3}{2}} n_1^{-\frac{3}{2}} \\
%&\approx 2^n \sum_{a \leq o(n)} a n^{-\frac{7}{2}} n^{-\frac{1}{2}} n^{-\frac{3}{2}} \\
%&= 2^n \sum_{a \leq o(n)} a n^{-\frac{11}{2}} \\
%&\preccurlyeq 2^n n^{-\frac{7}{2}}
%\end{align*}

\begin{align*}
E_1(thin_n, thick_n) &\approx \text{const}\,n \sum_{\frac{n_1}{2}\succcurlyeq b \succcurlyeq 1}
(4^n n^{-\frac{7}{2}})  \left(\frac{b^{-\frac{3}{2}}(n_1-b)^{-\frac{3}{2}}}{n_1^{-\frac{3}{2}}}\right).
\end{align*}

Above, $\text{const}\, n$ reflects a summation over $a$. Configurations with $n_3=\frac{n}{10}-a$ are counted in the first term; the fraction crossing from thin to thick upon the flip indicated in Fig.~\ref{fig_outerplanar} is given by the second term within the sum.

\begin{align*}
E_1(thin_n, thick_n) &\approx 4^n \text{const}\, n \left(\sum_{\frac{n_1}{2} \succcurlyeq b \succcurlyeq 1} n^{-\frac{7}{2}} b^{-\frac{3}{2}} \right) \\
 &\approx 4^n \text{const}\, n (n^{-\frac{7}{2}} n^{-\frac{1}{2}}) \\
 &\approx 4^n n^{-3}.
\end{align*}

Neglected terms, such as $a$ comparable to $n$, are down by a power $n^{-\frac{1}{2}}$ and have been dropped.

Putting the three calculations together, we conclude that the isoperimetric Cheeger constant $k$ satisfies $$k \preccurlyeq \frac{4^n n^{-3}}{4^n n^{-\frac{5}{2}}} = n^{-\frac{1}{2}}$$ 

This means that the valence normalized Cheeger constant $h$ appropriate to random walks \cite{Chung} satisfies $$h \preccurlyeq n^{-\frac{3}{2}}$$ and that the mixing time is at least $O(n^{3/2})$.  A mixing time of $\approx n$ corresponds to the usual graph theoretic notion of an  ``expander.'' Numerical investigation of
this model indicates that the actual mixing time is $O(n^2)$.

%------------------------------------------------------------------------------
\section{Analytical results for the off-lattice loop gas}
\label{Appendix:C}

\subsection{An off-lattice loop gas model}

In this appendix, we study analytical properties of an off-lattice loop gas
model. 
The basis states of the model are configurations of non-intersecting, indistinguishable
loops, identifying loop configurations related by isotopy. 
%One of the questions we are interested in is whether the surgery (self-surgery) 
%moves make the loop gas gapped.
Loop configurations with at most $N$ loops can be  represented by unlabeled 
rooted trees with at most $N$ nodes, excluding the root node.
Using the recursion relations of Ref.~\cite{Robinson:75}
the number of such trees ({\sl i.e.}, the number of loop configurations) for a fixed number of nodes $n$  is given by
\begin{equation}
C(n)=\frac{1}{n-1}\sum_{k=1}^{n-1} C(n-k)\sum_{m|k}mC(m) \,, 
\label{Configurations}
\end{equation}
where $n>1$, $C(1)=1$,  and ``$m|k$'' denotes all $m$ which are factors of $k$. 
 A similar expression exists  for the number of leaves (excluding the root) of 
 unlabeled rooted trees, 
 \begin{equation}
L(n) = \sum_{k=1}^{n-1}C(n-k)\sum_{m|k}L(m) \,
\label{Leaves}
\end{equation}
with $L(1)=1$.

The Hamiltonian of our off-lattice model acts locally by the three types of moves shown in
Fig.~\ref{fig:lg:moves}: 
a) The inflation move corresponds to creating or annihilating a loop. 
b) The surgery move is merging of two loops. 
c) The self-surgery move is a surgery move of a loop with itself. 
We define the Hamiltonian $H$ as a sum of projectors performing inflation, surgery, and
self-surgery moves such that for the ground state wave function $\ket{\psi_0}$ we have $H|\psi_0\rangle =0$. 
The ground state wave function then becomes an equal-weight superposition
of all loop configurations $l$
\begin{equation}
|\psi_0\rangle = \sum_{l}  |l\rangle \,.
\end{equation}

\begin{figure}[t]
\begin{center}
\includegraphics[width=0.85\columnwidth]{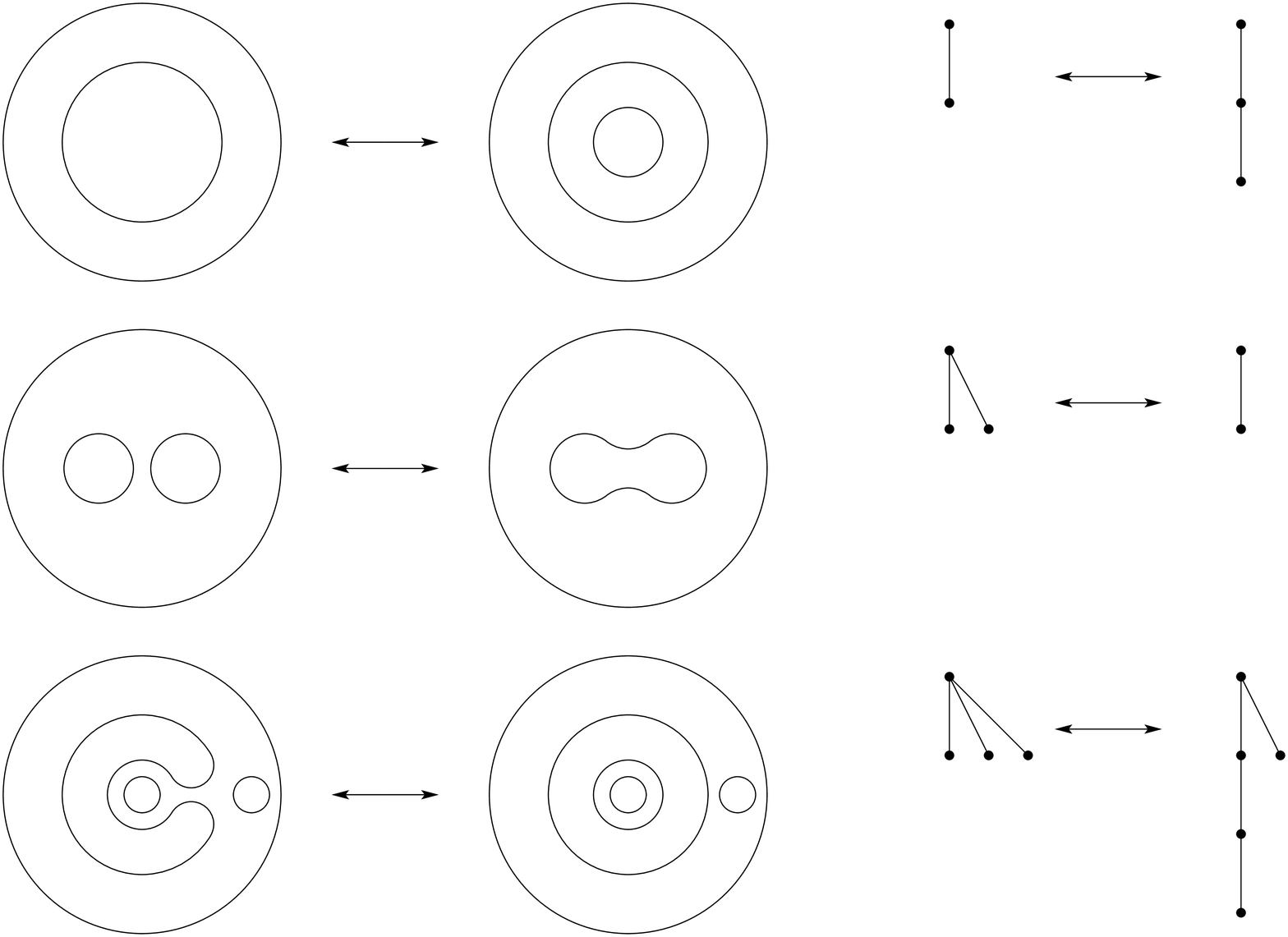}
\end{center}
\caption{Loop gas moves. From top to bottom: a) inflation move, b) surgery move,
and c) self-surgery move. }
\label{fig:lg:moves}
\end{figure}

Note that the Hamiltonian takes the form (up to rescaling) of an unweighted {\sl graph Laplacian} $L$ (see Appendix A): 
each transition (via inflation, surgery, of self-surgery moves) from  a state $|\alpha\rangle$
to a state $|\beta\rangle$ gives an entry of $-1$ in the Hamiltonian
matrix, and the diagonal elements are
$H_{\alpha\alpha}=-\sum_{\beta\ne\alpha} H_{\alpha\beta}$, {\sl i.e.}, the
diagonal element $H_{\alpha\alpha}$ count the number of transitions out of
state $|\alpha\rangle$.

%------------------------------------------------------------------------------
%\subsection{Mapping to unlabeled rooted trees}

%------------------------------------------------------------------------------
\subsection*{Topological protection}

We consider the   loop gas on an annulus (periodic boundaries in one direction).
A particular loop gas configuration can be represented by a tree where one leaf marks the inner edge of the annulus
and the root corresponds to the outer edge of the annulus as illustrated in Fig.~\ref{Fig:annulus}.
%Each configuration winds around the annulus either an even or an odd number of times.
The surgery move alters the number of loops that wind around the system  by $\pm 2$;
the parity of the winding is hence a conserved quantity.   

Is it possible to {\sl locally} distinguish even and odd winding sectors? 
In a lattice realization of a loop gas, such as the toric code \cite{ToricCode}, the expectation values of any local operator
in these sectors split by at most an exponentially small amount -- the hallmark of topological protection. 
In an off-lattice model, on the other hand, the splitting of these winding sectors turns out to be only algebraically small.
To see this, consider the average number of leaves $L^{p}_N$ in a sector with parity $p$. The difference between the
odd and even winding sectors
\begin{equation}
 A(N) =| \langle L^{\rm ev}_N \rangle - \langle L^{\rm odd}_N \rangle| \propto 1/N
\end{equation}
can be computed using Eqs.~\eqref{Leaves} and \eqref{Configurations} and is found to decay algebraically as $1/N$,
which is also illustrated in Fig.~\ref{leaves}.

\begin{figure}[ht]
  \begin{center}
  \includegraphics[width=0.6\columnwidth]{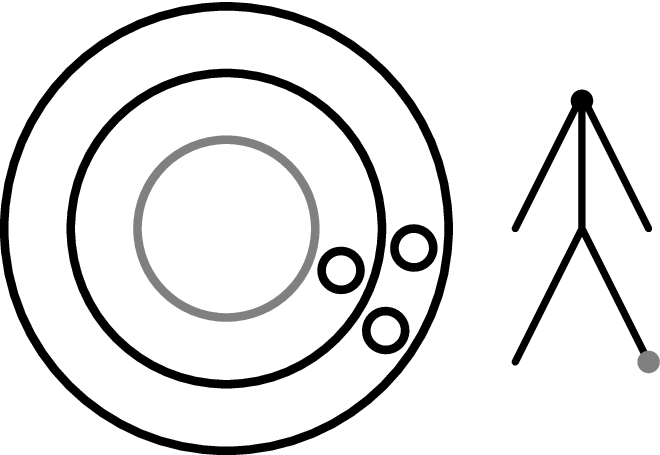}
  \end{center}
\caption{Loop gas on an annulus.}
\label{Fig:annulus}
\end{figure}

 \begin{figure}[ht]
\begin{center}  
\includegraphics[width=0.7\columnwidth]{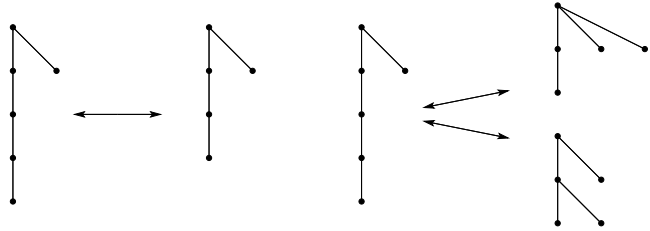}
\end{center}
\caption{ Height-changing moves (inflation and self-surgery --- surgery moves do not change the height) are only possible if a tree has only one leaf at its maximal height level, like the configuration shown here. Such configurations then allow for {\it one} inflation move and {\it one} self-surgery move 
 that changes the height by $1$ (left panel). In contrast, 
  there are many self-surgery moves possible that change the height by $2$
  (right panel).
}
\label{tree2}
\end{figure}

\begin{figure}[ht]
\begin{center}
\includegraphics[width=0.75\columnwidth]{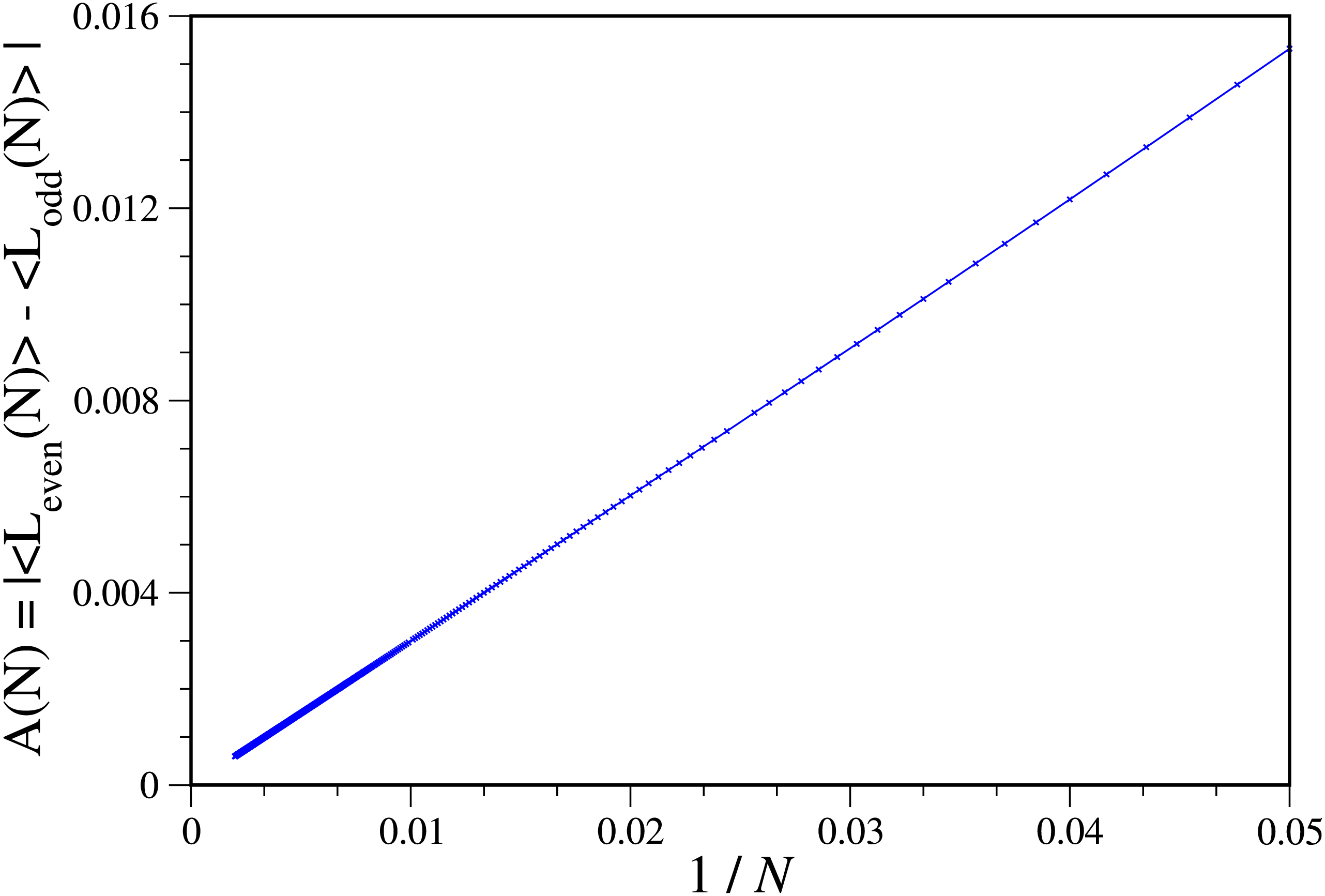}
\end{center}
\caption{
  Power-law decay of $A(N)$, the difference of the average number of leaves of unlabeled rooted trees in even and odd winding sectors 
  for an annulus. }
\label{leaves}
\end{figure}

%------------------------------------------------------------------------------
\subsection{Proof of gaplessness of loop gas Hamiltonian}

A loop gas configuration of $n$ loops can be represented by a rooted unlabeled
random tree with $n+1$ nodes. In the following, we shall refer to this tree
representation. The ground state of the  Hamiltonian is given by
\begin{equation}
|\psi_0\rangle = \frac{1}{\sqrt{C_N}} \sum_{\{\alpha\}} |\alpha \rangle
\end{equation}
where the sums runs over all possible rooted unlabeled random tree configurations, and $C_N$ is the total number
of such configurations
 with at most $N$ nodes,
\begin{equation}
C_N=\sum_{n=1}^N \sum_{h=0}^{N-1}C(n,h).
\end{equation}
where $C(n,h)$ denotes the number of tree configurations with $n$ nodes and height $h$ where $n>h$.
 The action of the Hamiltonian  is such that
\begin{equation}
H|\psi_0\rangle = 0.
\label{GS}
\end{equation} 
The terms of the loop gas Hamiltonian (inflation, surgery, self-surgery) 
are equivalent to the following modifications in the tree representation
(we list only the direction of the moves that remove a node):
inflation corresponds to removing a leaf of the tree, surgery corresponds to
``fusing" two sibling nodes into a single node, and self-surgery corresponds
to ``fusing" a child node with its grandparent node (all children nodes of
the child node become sibling nodes of its parent node, the whole process can
be visualized as ``folding" the tree such that child node and grandparent node
become one node).
The different moves in the Hamiltonian  change the tree height by either $0$
(all surgery moves, some self-surgery moves, some inflation moves), $1$ (some
inflation moves, some self-surgery moves), or $2$ (some self-surgery moves),
and thus
 \begin{eqnarray}
 -\langle \alpha |H|\alpha\rangle  &=& \sum_{\{\beta,h_{\beta}=h_{\alpha},\beta\ne \alpha\}} \langle \beta |H|\alpha\rangle
 + \sum_{\{\beta,h_{\beta}=h_{\alpha}+1\}} \langle \beta |H|\alpha\rangle
 + \sum_{\{\beta,h_{\beta}=h_{\alpha}-1\}} \langle \beta |H|\alpha\rangle\nonumber \\
 && + \sum_{\{\beta,h_{\beta}=h_{\alpha}+2\}} \langle \beta |H|\alpha\rangle
 + \sum_{\{\beta,h_{\beta}=h_{\alpha}-2\}} \langle \beta |H|\alpha\rangle.
 \label{Relation1}
 \end{eqnarray}
We define a ``constrained" number of configurations at height $h$,
\begin{equation}
 C(h) := \sum_{n,\ {\rm where} \ g(n) \ge h} C(n,h),
\end{equation}
where $g(n)$ is some function of $n$ to be
defined below. Using this definition, we make sure that only a constrained
number of configurations is included. Next, we consider the state
\begin{equation}
 |\psi_1\rangle =  \frac{1}{\sqrt{C}} \sum_{h=\bar{h} }^{ m \bar{h} -1}
  \frac{\exp(2\pi i h / \bar{h})} {C(h)}
  \sum_{\{\alpha,\ h_\alpha=h,\ g(n_\alpha) \ge h \}} |\alpha\rangle,
\label{eq:trial:state}
\end{equation}
where $m>1$, $\bar{h}={\rm Int}[k(N)]$ (${\rm Int}[x]$ denotes the smallest
integer number larger than $x$), and $C=\sum_{h=\bar{h}}^{m\bar{h}-1}1/C(h)$.
The following proof relies on the inequality $C(n,h+1)\le C(n,h)$ that
should be valid for all configurations included in the above trial
state. This inequality can be satisfied by proper choice of the functions
$g(n)$ and $k(n)$, see below.
   
The states  $|\psi_0\rangle$ and $|\psi_1\rangle$ are two distinct orthonormal
basis states,
\begin{eqnarray*}
 \langle \psi_1|\psi_1\rangle =
  \frac{1}{C}\sum_{h=\bar{h}}^{m\bar{h}-1} \frac{1}{(C(h) )^2}
    \sum_{\{\alpha,\ h_\alpha=h,\ g(n_\alpha) \ge h \}}
    \langle \alpha|\alpha\rangle =
  \frac{1}{C} \sum_{h=\bar{h}}^{m\bar{h}-1}\frac{1}{C(h) } = 1,
\end{eqnarray*}
and 
\begin{eqnarray*}
 \langle \psi_0|\psi_1\rangle &=&  \frac{1}{\sqrt{C C_N}}
   \sum_{h=\bar{h}}^{m\bar{h}-1} \frac{\exp(2\pi i h/\bar{h}) }{C(h)}  
    \sum_{\{\alpha,\ h_\alpha=h,\ g(n_\alpha) \ge h \}}
    \langle \alpha|\alpha\rangle \\
 &=& \frac{1}{\sqrt{C C_N}} \sum_{h=\bar{h}}^{m\bar{h}-1}
   \exp(2\pi i h/\bar{h}) = 0.
\end{eqnarray*}

The energy gap can be estimated as
\begin{eqnarray*}
\Delta E &\le& \langle \psi_1|H|\psi_1\rangle \\
&=& \frac{1}{C } \sum_{h=\bar{h}}^{m\bar{h}-1}\sum_{h'=\bar{h}}^{m\bar{h}-1}
 \frac{\exp( 2\pi i (h-h')/\bar{h}  )}{C(h)C(h')} 
\sum_{\{\beta,\ h_\beta=h',\ g(n_{\beta})\ge h' \}}
\sum_{\{\alpha,\ h_\alpha=h,\ g(n_{\alpha})\ge h\}} \langle \beta|H|\alpha \rangle .
\end{eqnarray*}
Using relation (\ref{Relation1}), and that $\langle \beta |H|\alpha\rangle < 0$
if $\alpha\ne \beta$, we obtain
\begin{eqnarray*}
 \Delta E &\le&
  \frac{1}{C }\sum_{h=\bar{h}}^{m\bar{h}-1}
    \left[ \frac{2}{(C(h))^2}-\frac{2\cos(2\pi/\bar{h}) }{C(h)C(h+1)} \right]
  \sum_{\{\beta,\ h_\beta=h+1,\ g(n_\beta) \ge h+1 \}}
    \sum_{\{\alpha,\ h_\alpha=h,\ n_\alpha=n_\beta-1 \}}
   |\langle \beta|H|\alpha\rangle|   \\
 &+& \frac{1}{C }\sum_{h=\bar{h}}^{m\bar{h}-1}
    \left[ \frac{2}{(C(h))^2}-\frac{2\cos(2\pi/\bar{h}) }{C(h)C(h+2)} \right]
  \sum_{\{\beta,\ h_\beta=h+2,\ g(n_\beta) \ge h+2 \}}
    \sum_{\{\alpha,\ h_\alpha=h,\ n_\alpha=n_\beta-1 \}}
   |\langle \beta|H|\alpha\rangle|.
 \end{eqnarray*}
A state $|\beta\rangle$ of height $h$ has at most one inflation move
transition to only one of all states $|\alpha\rangle$ of height $h-1$.
The same applies to self-surgery transitions that change the height by one.
A tree of height $h$ can have at most $h-2$ self-surgery transitions that
decrease the height by two. Using these estimates, and that $C(h+1)\le C(h)$
(since $C(n,h+1)\le C(n,h)$ for all configurations included in the trial
state), we obtain
\begin{eqnarray*}
  \Delta E &\le & \frac{1}{C} a_1 \sin^{2}(\pi/\bar{h})
   \sum_{h=\bar{h}}^{m\bar{h}-1} \left(
     \frac{2C(h+1)}{(C(h))^2} + \frac{(h-2)C(h+2)}{(C(h))^2}
   \right) \\
  &\le& \frac{1}{C} a_2 \sin^{2}(\pi/\bar{h}) \bar{h}m
    \sum_{h=\bar{h}}^{m\bar{h}-1} \frac{1}{C(h)}
  \le a_2 \sin^2(\pi/\bar{h}) m\bar{h} \sim \frac{a_3}{\bar{h}}
\end{eqnarray*}
for large $\bar{h}$, where $a_1$, $a_2$, and $a_3$ are constants.

\begin{figure}
\begin{center}
\includegraphics[width=10cm]{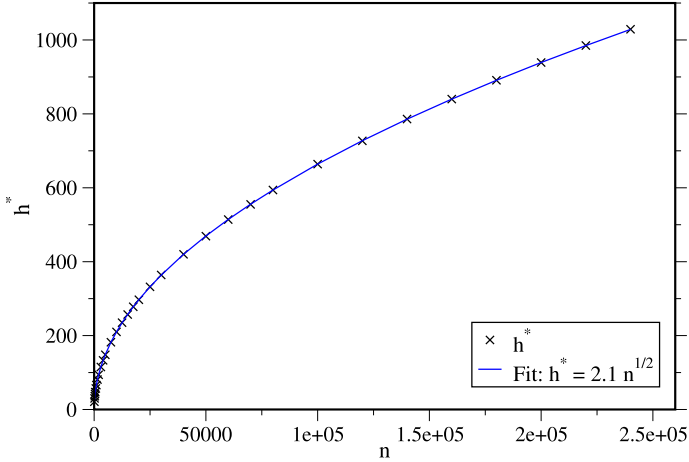}
\caption{Growth of $h^*$ (value of height where $C(n,h) $ is maximal at a given $n$)
 as a function of $n$ as obtained from (\ref{height_dist}). The growth is proportional to $\sqrt{n}$.}
\label{max_height}
\end{center}
\end{figure}

We need to choose the functions $g(n)$ and $k(n)$.
The authors of Ref.~\cite{Drmota:08} showed that for $n\to \infty$, 
\begin{equation}
C(n,h)\sim C(n) 2 b \sqrt{\frac{\rho\pi^5}{n}} \beta^4 \sum_{k\ge 1}k^2(2k^2\beta^2 \pi^2-3)\exp(-k^2\pi^2\beta^2),
\label{height_dist}
\end{equation}
where $\rho\approx 0.3383219$, $b\approx 2.6811266$,
$\beta=2\sqrt{n}/hb\sqrt{\rho}$, and\cite{Otter:48,Robinson:75}
\begin{equation*}
C(n) \sim \frac{b\sqrt{\rho}}{2\sqrt{\pi}} n^{-3/2}\rho^{-n}.
\end{equation*}
Eq.~(\ref{height_dist}) is asymptotically valid for arbitrary but fixed
$\delta$ and
$(\delta\sqrt{\log n})^{-1} \le h/\sqrt{n} \le \delta \sqrt{\log n}$.
It can be seen from Fig.~\ref{max_height} that the number of configurations
at a given $n$ is largest for $h^*(n)=2.1\sqrt{n}$. It follows from
Eq.~(\ref{height_dist}) that $C(n,h+1) \le C(n,h)$ in its region of validity
if $h\ge h^*(n)$. It is easy to check that $C(n,h+1) \le C(n,h)$ for
any configuration in Eq.~(\ref{eq:trial:state}) if we choose
$g(n)=\delta\sqrt{n\log n}$  and $k(n)= 2.1\sqrt{n}$.
Here, $\delta$ is chosen in such a way that there exists, for a given
$N$, at least one $n$ such that $\bar{h} \le h\le g(n)$.
Using  this choice of $g(n)$ and $k(n)$, we obtain the following estimate for
the energy gap
\begin{equation}
 \Delta E \le \frac{a_3}{\sqrt{N}},
\label{eq:tree:gap}
\end{equation}
where $a_3$ is a constant. 

It is likely that $C(n,h+1) \le C(n,h)$ as long as $h\ge h^*(n)$ (and not
only for $\delta\sqrt{n\log n} \ge h\ge h^*(n)$ as in the previous paragraph).
If this is true then we can derive a tighter bound for the gap.
Indeed, asymptotically $h\ge h^*(n)$ for any $h$ in Eq.~(\ref{eq:trial:state})
if we choose the following function $k(n)=an^{\kappa}$, where
$1/2\le\kappa\le1$ and $a$ is a constant. In this case, the gap scales
as $N^{-\kappa}$ and the upper bound is $N^{-1}$. The states with $\kappa$
that is close to 1 have support only on an exponentially small number of tree
configurations and thus detecting the $N^{-1}$ scaling in Monte Carlo
simulations seems unfeasible, see \ref{Appendix:Numerics}. 

Rescaling\footnote{We drop an exponentially small fraction of states, corresponding to some trees whose height scales slower than $O(\sqrt{n})$ and  whose connectivity scales faster than $N$} this Hamiltonian to a graph Laplacian $\cal L$  , we obtain a scaling of the gap as $N^{-3/2}$ (from Eq.~(\ref{eq:tree:gap})). Plugging this into the gap estimates for general local Hamiltonians in \ref{Appendix:Spectrum} we find an upper bound of $\sqrt{N}$. We will see below, in the numerical results of \ref{Appendix:Numerics} that the gap actually scales as $N^{-1.75}$, which is still not enough by itself to prove gaplessness of any model.

%------------------------------------------------------------------------------
\subsection*{Estimating Cheeger's constant}

 \begin{figure}[ht]
\begin{center}
\includegraphics[width=0.6\columnwidth]{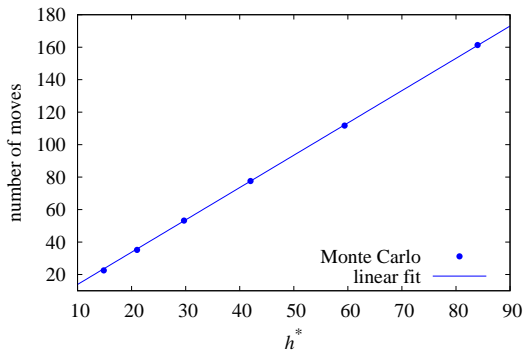}
\end{center}
\caption{Scaling of the average number of moves that take a tree across the mean height $h^*$. The perfect linear behavior shows that the number of moves indeed satisfies the linear upper bound.}
\label{fig:movecount}
\end{figure}

Since the gap of the graph Laplacian seems to be too weak a bound for the Cheeger constant, we next try to estimate the Cheeger constant directly. Using the same cut as in the above proof, we consider a split of the configuration space into those trees which are smaller or larger than the mean height $h^*$. The fraction of trees at this boundary can be obtained from equation (\ref{height_dist}) to be $C(n,h^*)/C(n) \sim 1/\sqrt{n}$. Multiplying this with the number of moves across the cut, which can be bounded by $O(h^*) \sim O(\sqrt{n})$ we obatin as estimate for Cheeger's constant
\begin{equation}
h \le O(1/\sqrt{N}) O(h^*(N)) \le O(1/\sqrt{N}) O(\sqrt{N}) = const. ,
\end{equation}
which is now border line regarding the absence of a gap. We hence tried to check numerically whether the number of height changing moves, which we have bounded by $h^*$ might grow slower then linear. However it turns out, as is shown in Fig. \ref{fig:movecount} that the scaling indeed satisfies this bound.

%------------------------------------------------------------------------------
\section{Numerical study of off-lattice loop gases and string nets}
\label{Appendix:Numerics}

%------------------------------------------------------------------------------

\begin{figure}[ht]
\begin{center}
\includegraphics[width=0.5\columnwidth]{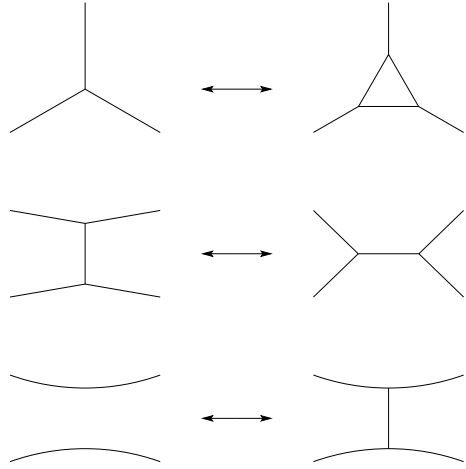}
\end{center}
\caption{String net moves. From top to bottom: inflation move, $f$-move,
and surgery move.}
\label{fig:sn:moves}
\end{figure}

In this appendix, we turn to a numerical analysis of the off-lattice loop gas
and string net models by Monte Carlo and exact diagonalization. 
The off-lattice loop gas model has been introduced in the previous appendix. 
In a similar fashion, off-lattice string nets can be defined as
indistinguishable (unlabeled) planar trivalent graphs, where we exclude
configurations with bubbles or parallel edges. We define the system size $N$
as the maximum number of faces, which is related to the number of vertices
$n_v$ via $N=(n_v+4)/2$. The Hamiltonian again takes the form of a graph
Laplacian and is defined by the three types of moves illustrated in
Fig.~\ref{fig:sn:moves}. The string net ground state again is an equal-weight
superposition of all string net configurations $s$
$$
  |\psi_0 \rangle = \sum_{s} |s\rangle.
$$

\begin{figure}[t]
\begin{center}
\includegraphics[width=0.6\columnwidth]{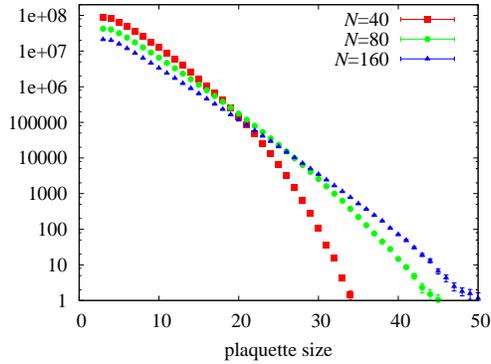}
\end{center}
\caption{Distribution of the plaquette sizes in the off-lattice string net}
\label{fig:plaquettesize}
\end{figure}

The definition of the graph Laplacian ${\cal L}$ in \ref{Appendix:Spectrum}
includes the constant $c=\min_i(d_i/\sum_{j\neq i}|H_{ij}|)$. The sum in
the denominator is basically the (weighted) vertex degree of vertex
(basis state) $i$. It typically grows faster than the system size $N$.
However, the number of vertices (or basis states in the Hilbert space) for
which the weighted vertex degree is not bounded by a linear function of $N$
is exponentially small. As an example, surgery moves within a plaquette of a string net grow like the square of the the number of edges in the plaquette, since one can do surgery between any pair of edges. Since plaquettes with a large number of edges are exponentially suppressed (see Fig. \ref {fig:plaquettesize}), we  discard those exponentially
rare states and restrict the Hilbert space ${\cal H}$ to ${\cal H}'$ such
that a basis state $|i'\rangle$ belongs to ${\cal H}'$ iff
$|i'\rangle\in {\cal H}$ and $\sum_{j\neq i'}|H_{i'j}|/d_{i'}$ is bounded by
a linear function of $N$. The graph Laplacian ${\cal L}'$ is then defined following \ref{Appendix:Spectrum}, with $c'$  proportional to $N^{-1}$.

The graph Laplacian is gapless by definition. In the following, we also
demonstrate the gaplessness of $N{\cal L}'$ for both the off-lattice loop gas
and string net models by numerically determining the gap to the first excited
mode of ${\cal L}'$.

%------------------------------------------------------------------------------
\subsection*{Monte Carlo method}

One can extract the gap of the graph Laplacian from classical Monte Carlo
simulations \cite{henley} by ensuring that the Monte Carlo transition matrix
is proportional to the graph Laplacian ${\cal L}'$:
\begin{equation}
 T = (1 - \alpha) {\mathbb I} + \alpha {\cal L}',
\label{eq:th}
\end{equation}
where $T$ is the transition matrix and $\alpha$ is the coefficient of
proportionality.

We perform Monte Carlo simulations by first calculating the number
$N_\text{moves}$ of possible moves for a given configuration.
Let $\widetilde{N}_\text{moves}=1/c'$. Then $\widetilde{N}_\text{moves}$
is a linear function of $N$ and it is larger or equal to $N_\text{moves}$
for any configuration. We randomly pick one of the possible moves and
accept it with probability $N_\text{moves}/\widetilde{N}_\text{moves}$.
The total probability to make a move is $1/\widetilde{N}_\text{moves}$ and it
is the same for any move. The probability to stay in the given configuration
is $1-N_\text{moves}/\widetilde{N}_\text{moves}$. This transition matrix is
equal to the graph Laplacian ${\cal L}'$ ($\alpha=1$).

The enumeration of possible moves can be implemented very efficiently for
rooted trees allowing us to access large system sizes. However, this is not
the case for the string nets as one needs to check for graph isomorphisms
for every possible move, which restricts us to considerably smaller sizes. 
We use the isomorphism test suggested in Ref.~\cite{plantri}.

The gap is related to the autocorrelation time $\tau_A$ of some observable
$A$ as
$$
  \Delta=1-e^{-1/\tau_A},
$$
where $\tau_A$ is measured in Monte Carlo time.
The observable $A$ must be chosen carefully -- it must couple to the lowest 
mode in order to extract the gap.

%------------------------------------------------------------------------------
\subsection*{Exact diagonalization results}

\begin{figure}[t]
\begin{center}
\includegraphics[width=0.49\columnwidth]{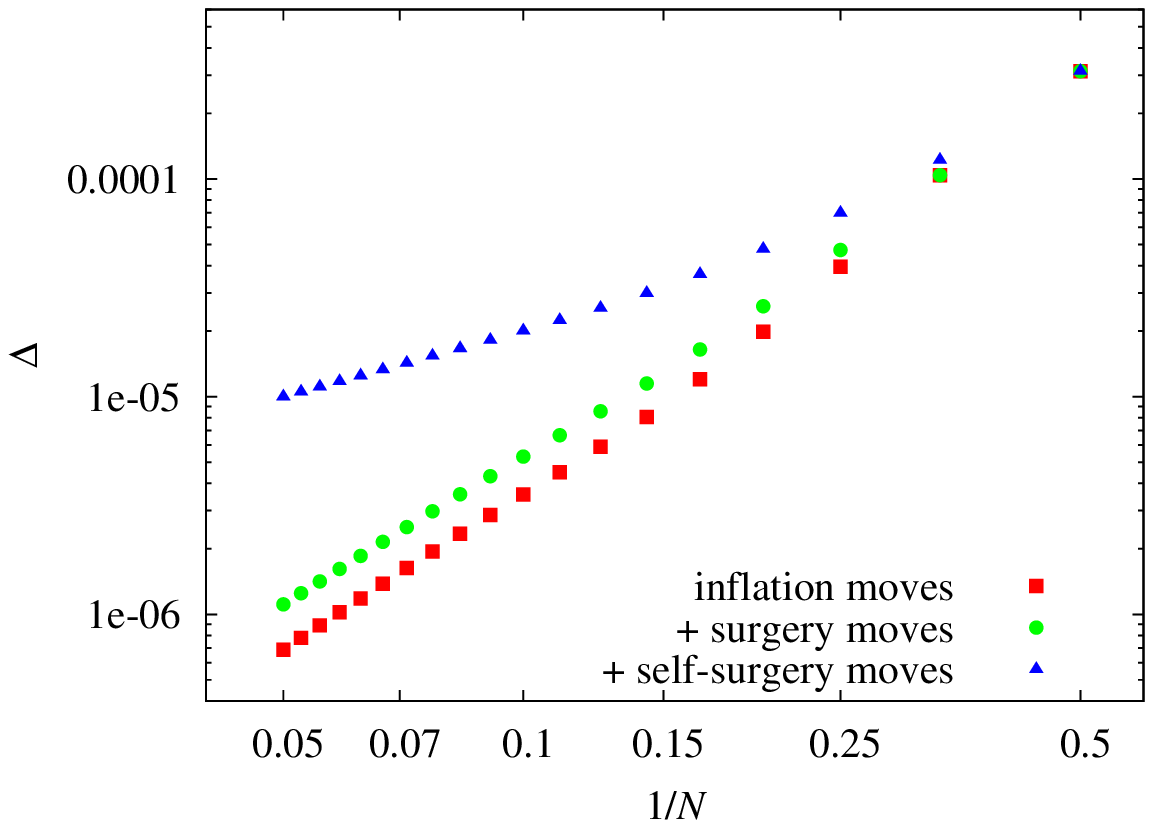}
\includegraphics[width=0.49\columnwidth]{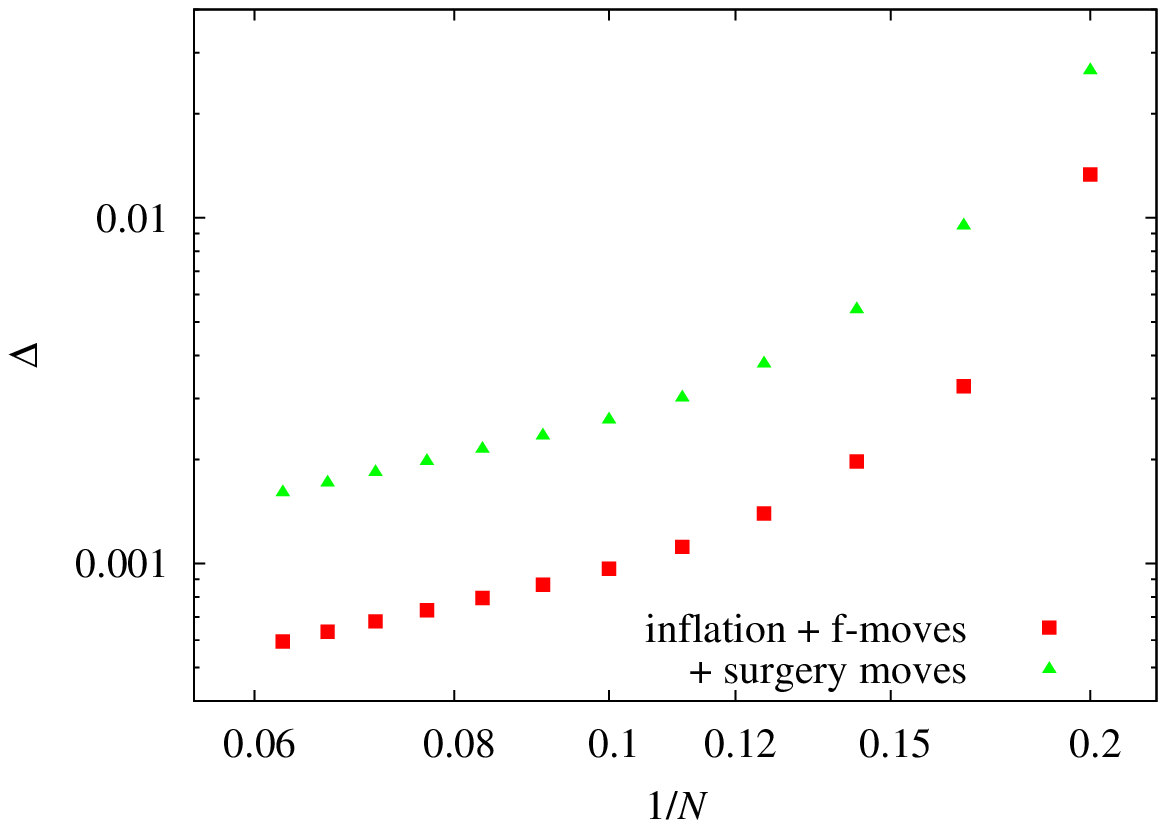}
\end{center}
\caption{
{\sl Exact diagonalization results:}
The gap of the graph Laplacian as a function of the inverse
system size $1/N$ obtained from exact diagonalization. Results for
the off-lattice loop gas are shown on the left ($1/c'=3200N$), and for
the off-lattice string net on the right ($1/c'=30N$).}
\label{fig:ED-GraphLaplacian}
\end{figure}

We first analyze the spectrum of the graph Laplacian for off-lattice loop gas
and string net model using exact diagonalization. In particular, we calculate
the lowest gap using the Lanczos algorithm \cite{Lanczos} as shown in 
Fig.~\ref{fig:ED-GraphLaplacian}.
For the off-lattice loop gas, we find that if we only consider inflation and
surgery moves, the graph Laplacian times the system size is clearly gapless
-- consistent with the proof in \ref{Appendix:C}.
Adding self-surgery moves the gap of $N{\cal L}'$ appears to extrapolate
to a {\sl finite} value. The same is seen for the off-lattice string net.
%However, the exact diaganolization approach is limited to relatively small loop gases 
%with up to 20 loops or string nets with up to $N=16$.
However, this apparent convergence is misleading as we will see below in
Monte Carlo simulations of larger systems.

%------------------------------------------------------------------------------
\subsection*{Monte Carlo results}
\label{Appendix:NumericsSphere}

\begin{figure}[t]
a) loop gas \\
\begin{center}
  \includegraphics[width=0.49\columnwidth]{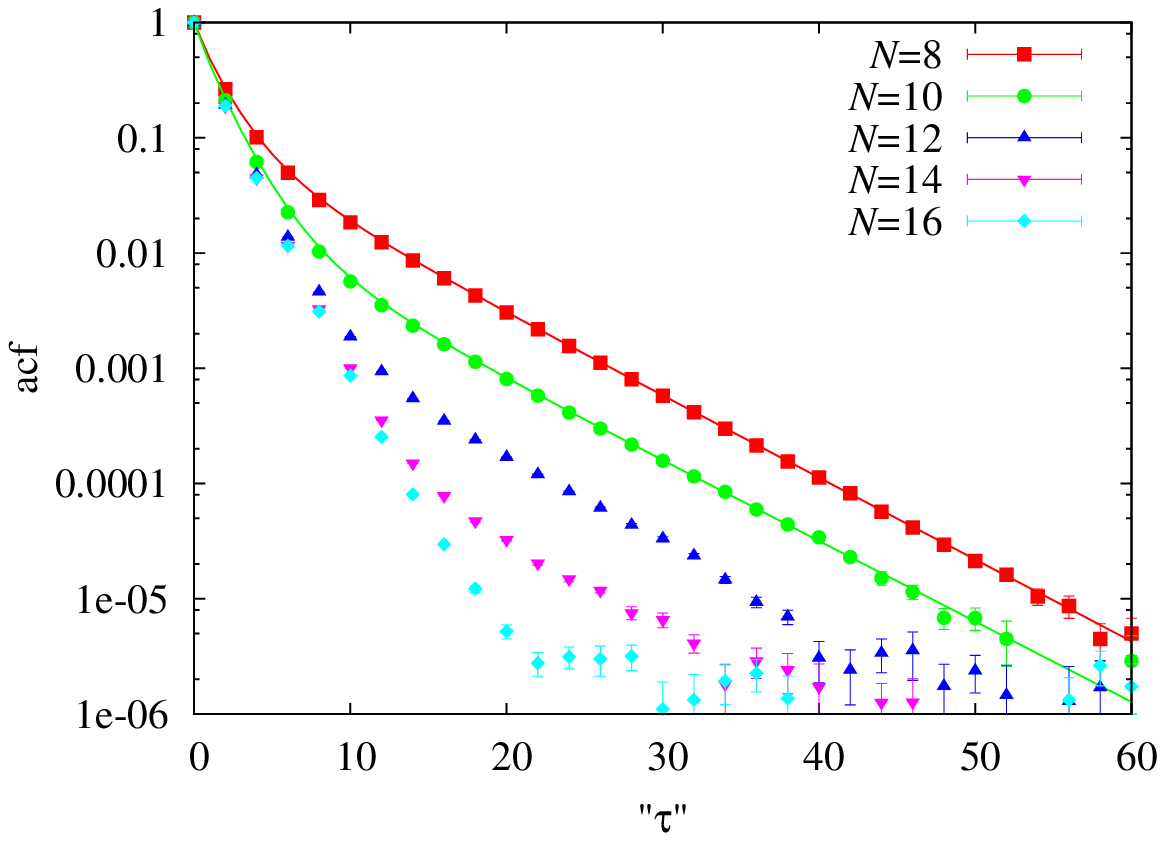}
  \includegraphics[width=0.49\columnwidth]{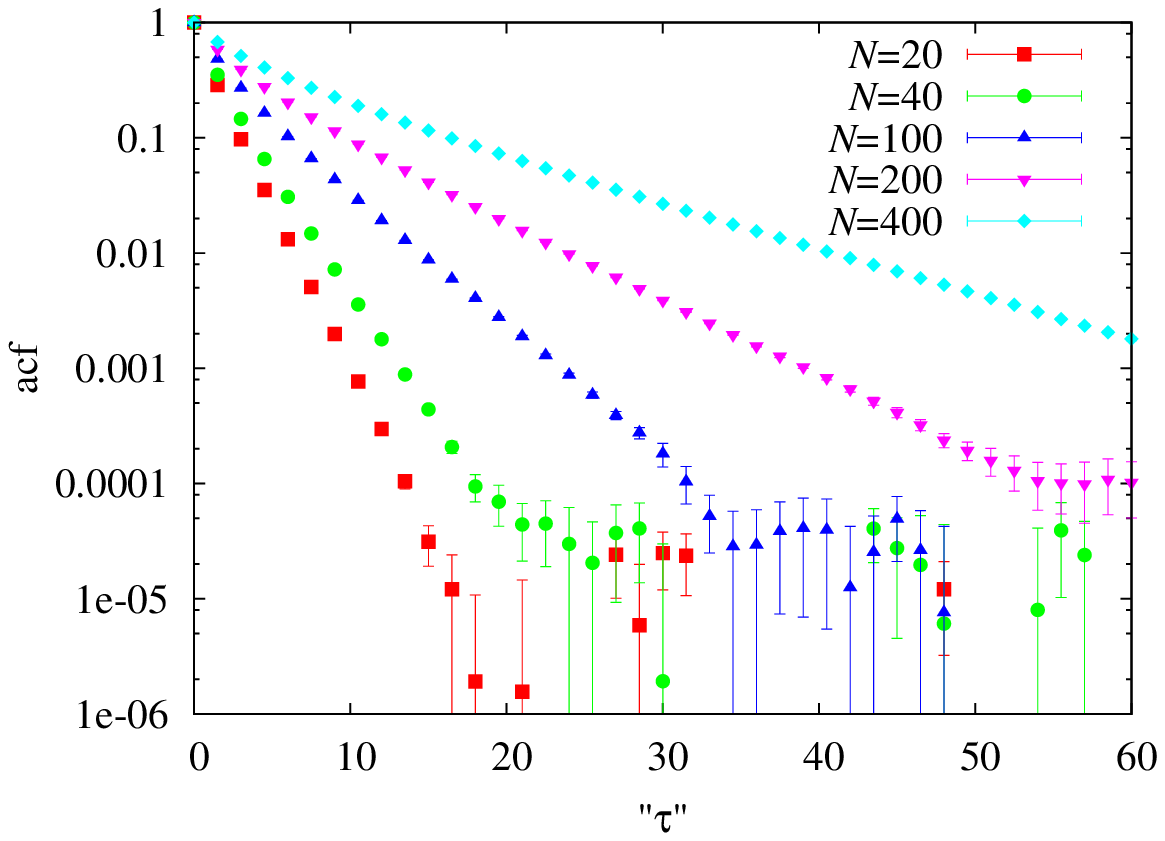}
\end{center}
b) string net \\
\begin{center}
  \includegraphics[width=0.49\columnwidth]{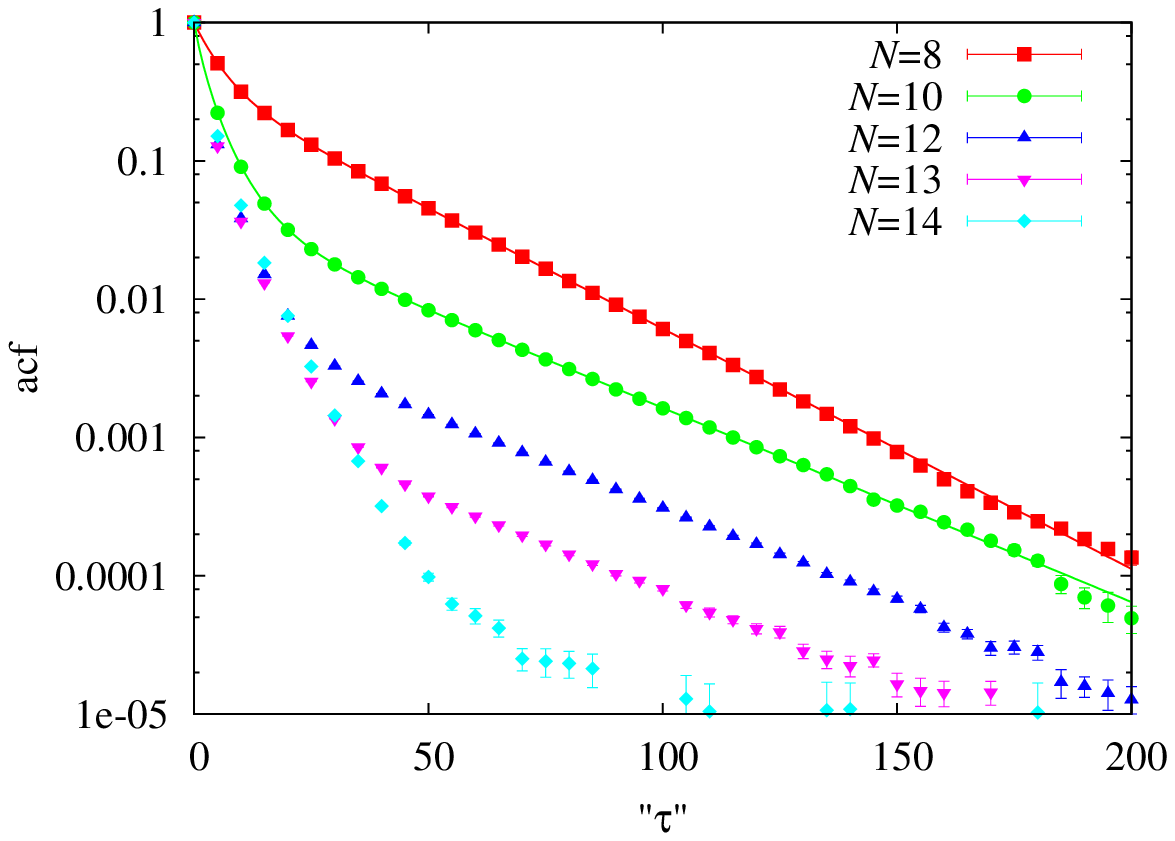}
  \includegraphics[width=0.49\columnwidth]{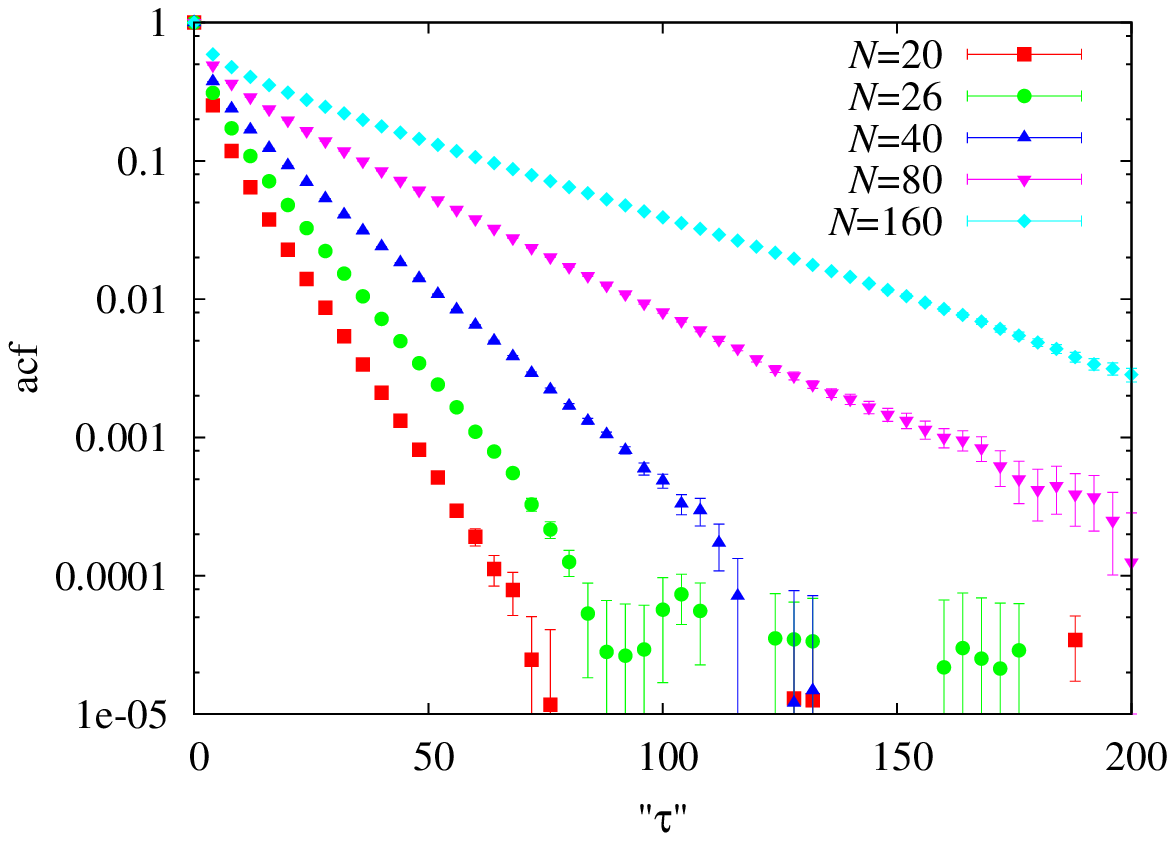}
\end{center}
\caption{
{\sl Autocorrelation functions:}
Monte Carlo results for the autocorrelation function of
a) the tree height for the loop gas
b) graph diameter for the string nets.
The Monte Carlo time $\tau$ is given in arbitrary units. 
Lines denote exact diagonalization results.}
\label{fig:lg:acf}
\end{figure}

\begin{figure}[t]
\begin{center}
\includegraphics[width=0.49\columnwidth]{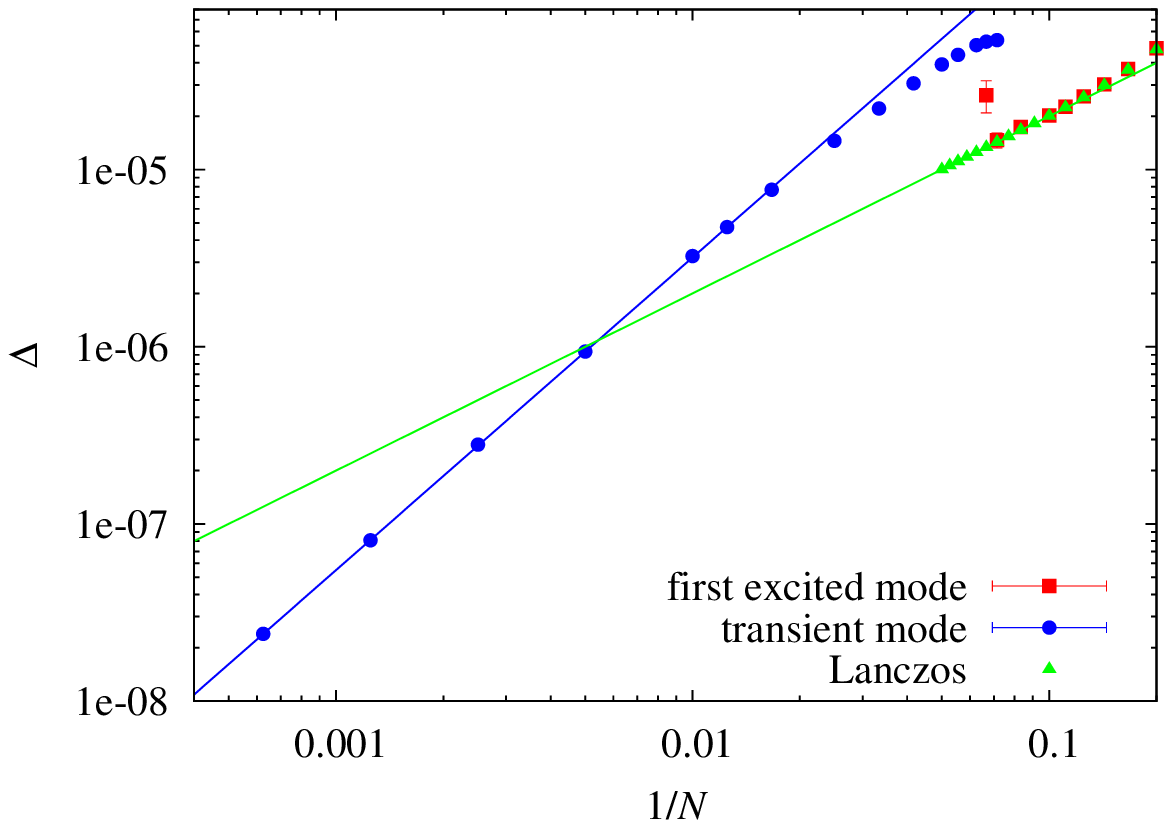}
\includegraphics[width=0.49\columnwidth]{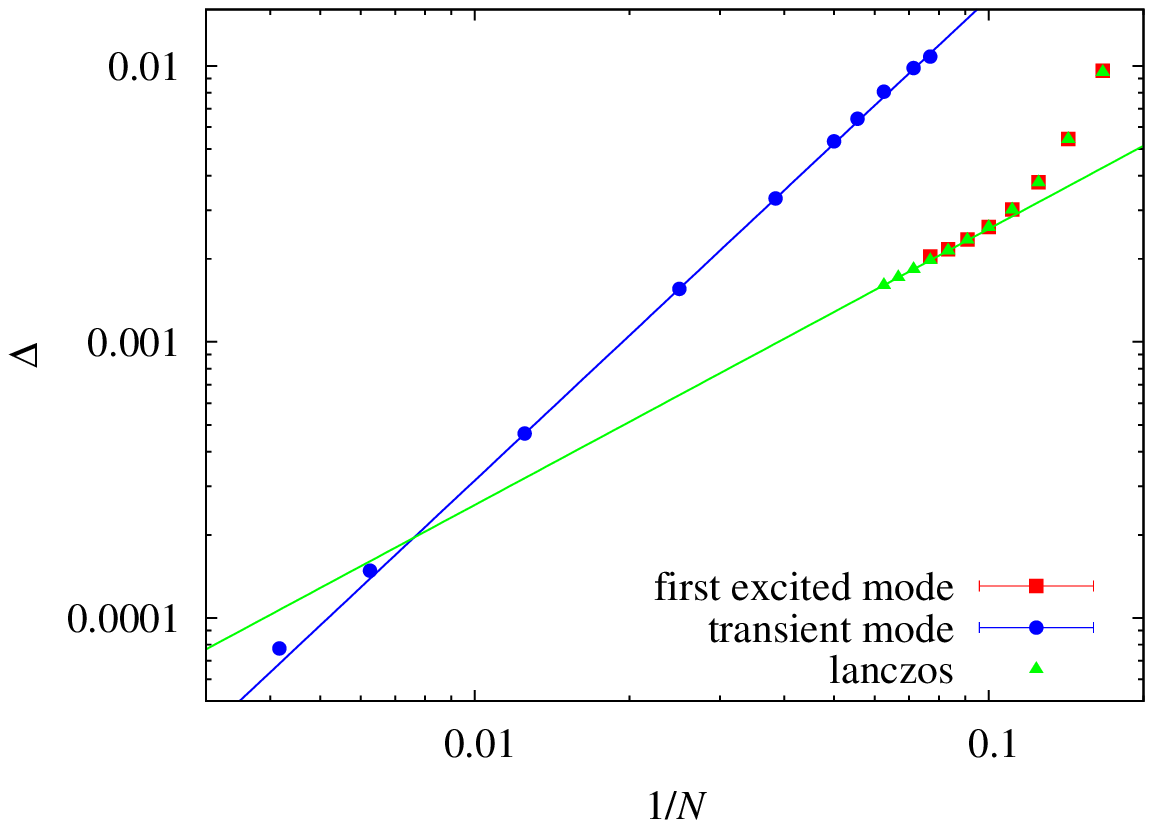}
\end{center}
\caption{
{\sl The gap of the graph Laplacian} as a function of
the inverse system size $1/N$ obtained from exact diagonalization and Monte
Carlo simulations. Results for the off-lattice loop gas are shown on the left
($1/c'=3200N$), and for the off-lattice string net on the right ($1/c'=30N$).
The finite-size extrapolation for large system sizes reveals that the gap
closes as $N^{-1.765(6)}$ for the loop gas and as $N^{-1.746(4)}$ for
the string net.}
\label{fig:lg:gap2}
\end{figure}

%\begin{figure}[t]
%\begin{center}
%\includegraphics[width=0.49\columnwidth]{lg_gap_fit}
%\includegraphics[width=0.49\columnwidth]{sn_gap_fit}
%\end{center}
%\caption{
%{\sl  The gap of the graph Laplacian:} The finite-sitze extrapolation for
%large system sizes reveals that the gap closes as }
%\label{fig:lg:gapfit}
%\end{figure}

To determine the gap of the graph Laplacian in Monte Carlo simulations we measure
the autocorrelation function of the tree height (for the loop gas) or graph diameter 
(for the string net). 
As shown in Fig.~\ref{fig:lg:acf} we find that, for small system sizes, the autocorrelation 
functions couple to high energy modes resulting in a fast initial decay before turning to 
a slower asymptotic behavior corresponding to the smallest gap.
This makes it difficult to extract the gap for large system sizes, since at long times the 
autocorrelation function is very small and noisy. 
To overcome this obstacle, we then fit the autocorrelation function to the transient behavior
at intermediate times, which will {\sl overestimate} the gap, thereby providing an upper
bound.

As shown in Fig.~\ref{fig:lg:gap2} the gap obtained from the asymptotic
behavior for small system sizes agrees perfectly with the exact
diagonalization results. For intermediate system size the transient behavior
overestimates the gap. However, for very large system sizes we see that this
upper bound goes to zero with increasing system size faster than the gap
extrapolated from the exact diagonalization results.
Fitting the large-$N$ behavior to a power-law $N^{-1-z}$ we obtain
$z=0.765(6)$ for the loop gas, and $z=0.746(4)$ for the string net.
The graph Laplacian times the system size ($N{\cal L}'$) is hence {\sl gapless}.

There is a simple heuristic argument for the crossover scale between the gapped behavior
for small $N$ and the gapless behavior for large $N$ in the case of the string net. 
This argument is best discussed in the dual picture of triangulations of the sphere.
For a small number of triangles the geometry is always that of a simple sphere and the updates
mix well resulting in gapped behavior. For $N$ larger than about 40, one can -- for the first time --
find triangulations that correspond to a geometry of two spheres described by two icosahedra 
connected by a narrow neck.
Updates no longer mix well, in particular there is a slow mode associated with shifting triangles
from one sphere to the other via the narrow neck. It is this slow mode which dominates the mixing
times for large system sizes resulting in gapless behavior.

%------------------------------------------------------------------------------
\subsection*{Outer planar triangulations}

\label{Appendix:N-gons}

\begin{figure}[t]
\begin{center}
\includegraphics[width=0.6\columnwidth]{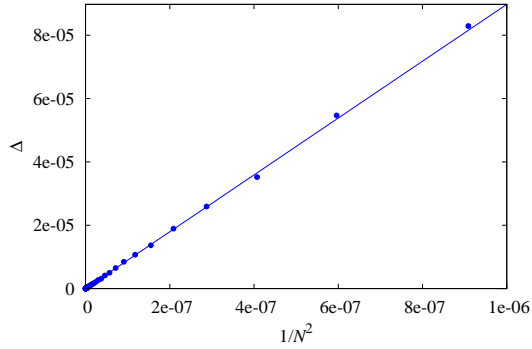}
\end{center}
\caption{The gap of the graph Laplacian for outer planar triangulations
($1/c'=N$).}
\label{fig:ngon}
\end{figure}

Finally, in Fig. \ref{fig:ngon} we show Monte Carlo results for the gap of the graph Laplacian for outer planar triangulations as a function of the number of triangles $N$. The observed decrease as $N^{-2}$ is consistent with but faster than the bound $N^{-3/2}$ derived in \ref{Appendix:OuterPlanar}. This scaling is fast enough to marginally show the gaplessness of any local Hamiltonian based on $F$ moves for this model.

\section*{References}
\bibliographystyle{elsarticle-num}
%\bibliography{graphlaplacians}

\end{document}